\documentclass{lmcs}

\usepackage{amssymb}

\usepackage[xcolor,hyperref,notion,quotation,electronic]{knowledge}

\usepackage{listings}
\newcommand{\iNTsim}{""NT_{\siml}""}
\newcommand{\NTsim}{"NT_{\siml}"}
\newcommand{\iRappl}{""R_{\approxl}""}
\newcommand{\Rappl}{"R_{\approxl}"}
\newcommand{\iRsim}{""R_{\siml}""}
\newcommand{\Rsim}{"R_{\siml}"}
\newcommand{\itgt}{""\theta""}
\newcommand{\tgt}{"\theta"}

\newcommand{\iapproxl}{\mathbin{""\approx_{L}""}}
\newcommand{\approxl}{\mathbin{"\approx_{L}"}}
\newcommand{\napproxl}{\mathbin{"\not\approx_{L}"}}
\newcommand{\inapproxsbot}{{""\not\approx_{S}\!\bot""}}
\newcommand{\napproxsbot}{{"\not\approx_{S}\!\bot"}}

\newcommand{\isfl}[1][L]{""\textit{sf}""_{#1}}
\newcommand{\sfl}[1][L]{"\textit{sf}"_{#1}}

\newcommand{\iLs}{""L^{\textit{sf}}""}
\newcommand{\Ls}{"L^{\textit{sf}}"}
\newcommand{\isiml}{\mathbin{""\sim_L""}}
\newcommand{\siml}{\mathbin{"\sim_L"}}
\newcommand{\nsiml}{\mathbin{"\not\sim_L"}}
\newcommand{\isims}{\mathbin{""\sim_S""}}
\newcommand{\sims}{\mathbin{"\sim_S"}}
\newcommand{\insims}{\mathbin{""\not\sim_S""}}
\newcommand{\nsims}{\mathbin{"\not\sim_S"}}
\newcommand{\iins}{\mathbin{""\in_S""}}
\newcommand{\ins}{\mathbin{"\in_S"}}
\newcommand{\inins}{\mathbin{""\not\in_S""}}
\newcommand{\nins}{\mathbin{"\not\in_S"}}
\newcommand{\ineqs}{\mathbin{""\not=_S""}}
\newcommand{\neqs}{\mathbin{"\not=_S"}}
\newcommand{\iapproxbot}{\,""\approx_L\!\bot""}
\newcommand{\approxbot}{\,"\approx_L\!\bot"}

\newcommand{\napproxbot}{\,"\not\approx_L\!\bot"}
\newcommand{\iapproxd}{\mathbin{""\approx_{\Dd}""}}
\newcommand{\approxd}{\mathbin{"\approx_{\Dd}"}}
\newcommand{\napproxd}{\mathbin{"\not\approx_{\Dd}"}}
\newcommand{\inapproxs}{\mathbin{""\not\approx_{S}""}}
\newcommand{\napproxs}{\mathbin{"\not\approx_{S}"}}
\newcommand{\iAaCc}{""\mathcal{A}[\mathcal{C}]""}
\newcommand{\AaCc}{"\mathcal{A}[\mathcal{C}]"}

\newcommand{\es}{\emptyset}
\newcommand{\imp}{\Rightarrow}
\renewcommand{\iff}{\Leftrightarrow}

\newcommand{\incl}{\subseteq}

\newcommand{\prefeq}{\mathbin{"\sqsubseteq"}}

\newcommand{\sq}[1]{[#1]}
\newcommand{\di}[1]{\langle #1 \rangle}
\newcommand{\sqq}[1]{\sq{\cdot }}
\newcommand{\ddi}[1]{\di{\cdot }}
\newcommand{\set}[1]{\{#1\}}

\renewcommand\bar[1]{\overline{#1}}

\newcommand{\act}[1]{\stackrel{#1}{\longrightarrow}}

\newcommand{\iasize}{""\alpha""}
\newcommand{\asize}{"\alpha"}

\renewcommand{\b}{\beta}

\newcommand{\e}{\varepsilon}

\newcommand{\g}{\gamma}
\renewcommand{\l}{\lambda}

\newcommand{\w}{\omega}

\renewcommand{\S}{\Sigma}

\newcommand{\Aa}{\mathcal{A}}
\newcommand{\Bb}{\mathcal{B}}
\newcommand{\Cc}{\mathcal{C}}
\newcommand{\Dd}{\mathcal{D}}

\newcommand{\Oo}{\mathcal{O}}

\newcommand{\struct}[1]{\langle #1 \rangle}

\def\sqr#1#2{\vbox
 {\hrule height#2
  \mbox{\vrule width#2 height#1 \kern#1 \vrule width#2}%
  \hrule height#2}}

\newcommand{\aact}[1]{\xrightarrow{#1}}
\newcommand{\bact}[1]{\xmapsto{#1:1}}
\renewcommand{\act}[1]{\xrightarrow{#1:2}}

\newcommand{\ecls}[1]{[#1]_{\tbe}}

\newcommand{\bv}{\bar v}

\newcommand{\init}{\mathit{init}}
\newcommand{\iFindR}{""\text{\emph{Find-R}}""}
\newcommand{\FindR}{"\text{\emph{Find-R}}"}

\newcommand{\iFindNT}{""\text{\emph{Find-NT}}""}
\newcommand{\FindNT}{"\text{\emph{Find-NT}}"}
\newcommand{\iFindu}{""\text{\emph{Find-u}}""}
\newcommand{\Findu}{"\text{\emph{Find-u}}"}
\newcommand{\iFindx}{""\text{\emph{Find-x}}""}
\newcommand{\Findx}{"\text{\emph{Find-x}}"}
\newcommand{\iLearn}{""\text{\emph{Learn}}""}
\newcommand{\Learn}{"\text{\emph{Learn}}"}
\newcommand{\iconstruct}{""\text{\emph{Construct}}""}
\newcommand{\construct}{"\text{\emph{Construct}}"}

\newcommand{\iautomaton}{""\text{\emph{Automaton}}""}
\newcommand{\automaton}{"\text{\emph{Automaton}}"}
\newcommand{\iFindPointedL}{""\text{\emph{FindPointed}}""}
\newcommand{\FindPointedL}{"\text{\emph{FindPointed}}"}
\newcommand{\iFindPointedS}{""\text{\emph{FindPointedInS}}""}
\newcommand{\FindPointedS}{"\text{\emph{FindPointedInS}}"}

\newcommand{\tbe}{\mathbin{"\equiv_L"}}
\newcommand{\ntbe}{\mathbin{"\not\equiv_L"}}
\newcommand{\itbe}{\mathbin{""\equiv_L""}}
\newcommand{\iAae}{""\mathcal{A}_{\equiv_L}""}
\newcommand{\Aae}{"\mathcal{A}_{\equiv_L}"}

\knowledge{notion}
  | \text {\emph {FindPointed}}

\knowledge{notion}
 | safe language

\knowledge{notion}
  | \text {\emph {FindPointedInS}}

\knowledge{notion}
  | \text {\emph {Automaton}}

\knowledge{notion}
| \text {\emph {Find-u}}

\knowledge{notion}
| \text {\emph {Find-x}}

\knowledge{notion}
| \text {\emph {Construct}}

\knowledge{notion}
| \not \approx _{S}\!\bot

\knowledge{notion}
| \text {\emph {Find-R}}

\knowledge{notion}
  | \text {\emph {Find-NT}}

\knowledge{notion}
  | \text {\emph {Learn}}

\knowledge{notion}
| \theta

\knowledge{notion}
| \alpha

\knowledge{notion}
  | R_{\approxl }

\knowledge{notion}
  | R_{\siml }

\knowledge{notion}
 | NT_{\siml }

\knowledge{notion}
  | \sim _L
  | \not\sim _L

\knowledge{notion}
  | \sim_S

\knowledge{notion}
  | \not\sim_S

\knowledge{notion}
  | \in_S

\knowledge{notion}
  | \not\in_S

\knowledge{notion}
  | \not=_S

\knowledge{notion}
| \approx _{L}
| \not\approx _{L}

\knowledge{notion}
| \approx_{\Dd}
| \not\approx_{\Dd}

\knowledge{notion}
  | \not\approx_{S}

\knowledge{notion}
  | \approx _L\!\bot 
  | \not \approx _L\!\bot 
 
\knowledge{notion}
  | \textit {sf}

  \knowledge{notion}
  | L^{\textit {sf}}

\knowledge{notion}
  | \approxbot 

\knowledge{notion}
  | \mathcal{A}[\mathcal{C}]

\knowledge{notion}
  | smaller
  | smallest

\knowledge{notion}
  | component
  | components

\knowledge{notion}
  |co-Büchi automaton
  |co-Büchi

\knowledge{notion}
  | supported

\knowledge{notion}
  | double-supported

\knowledge{notion}
  | single-supported

\knowledge{notion}
  | R_L

\knowledge{notion}
  | E_L

\knowledge{notion}
	| normalized
	| normal
  | normality

\knowledge{notion}
	| semantically-deterministic
	| semantic-determinism
  | Semantic-determinism

\knowledge{notion}
	| safe-deterministic
	| safe-determinism

\knowledge{notion}
	| safe SCCs
	| safe SCC

\knowledge{notion}
	| history-deterministic
	| history-determinism
        | hd
        | HD

\knowledge{notion}
	| unsafe-saturated
	| unsafe-saturation

\knowledge{notion}
	| central sequence
	| central sequences
	| central

\knowledge{notion}
	| pointed
	| pointed pair

\knowledge{notion}
	| After
	| \mathit{After}

\knowledge{notion}
	| \equiv_L
  | \not\equiv_L

\knowledge{notion}
	| \mathcal{A}_{\equiv}
	| \mathcal{A}_{\equiv_L}

\knowledge{notion}
	| \approx

\knowledge{notion}
	| \sqsubseteq
        
\knowledge{notion}
	| \approx\bot

\knowledge{notion}
	| \mathit{tgt}

\knowledge{notion}
        | samples
        | sample

\knowledge{notion}
        | characteristic sample
        | characteristic

\knowledge{notion}
        | consistent

\knowledge{notion}
        | extension
        | extends
        | extending

\knowledge{notion}
        | passive learner
        | learner

\knowledge{notion}
        | consistent learner

\knowledge{notion}
        | polynomial data

\knowledge{notion}
        | $L$-stabilizing
        | stabilizing

\knowledge{notion}
        | strategy


\begin{document}

\title[Minimal History-Deterministic Co-Büchi Automata]{Minimal History-Deterministic Co-Büchi Automata: Congruences and Passive Learning
} 
\thanks{Christof Löding is supported by DFG grant LO 1174/7-1.}
\thanks{This a long version of the LICS'25 article.}

\author[C. Löding]{Christof Löding\lmcsorcid{0000-0002-1529-2806}}[a]
\author[I. Walukiewicz]{Igor Walukiewicz\lmcsorcid{0000-0001-8952-7201}}[b]

\address{RWTH Aachen University, Aachen, Germany}
\address{CNRS, Bordeaux University, Talence, France}

\begin{abstract}
Abu Radi and Kupferman (2019) demonstrated the efficient minimization of history-deterministic (transition-based) co-Büchi automata, building on the results of Kuperberg and Skrzypczak (2015).
We give a congruence-based description of these minimal automata, and a self-contained
proof of its correctness.
We use this
description based on congruences to create a passive learning algorithm that can
learn minimal history-deterministic co-Büchi automata from a set of labeled example
words. The algorithm runs in polynomial time on a given set of examples, and
there is a characteristic set of examples of polynomial size for each minimal
history-deterministic co-Büchi automaton. 
\end{abstract}

\maketitle

\section{Introduction} \label{sec:intro}
Automata on infinite words, called $\omega$-automata, have been studied since
the early 1960s as a tool for solving decision problems in logic \cite{Buchi62}
(see also \cite{Thomas90}). Algorithms in formal verification of systems use
various types of automata; for example, nondeterministic $\omega$-automata are
used in model-checking procedures \cite{BaierK2008} while deterministic
$\omega$-automata find applications in the verification of probabilistic
systems~\cite{BaierK2008} or in the synthesis of reactive systems from
specifications (see \cite{Thomas09,LuttenbergerMS20} for a survey and some more
recent work). Despite this considerable interest, we still do not know how to
canonize $\w$-automata or how to learn them. 

Deterministic $\w$-automata are not a particularly promising model for canonization or learning.
For deterministic $\w$-automata, no small canonical forms are known.
The minimisation problem is NP-hard for all kinds of deterministic
$\omega$-automata expressive enough to recognise all $\omega$-regular
languages~\cite{Schewe10,Casares22,Abu.Ehl.NPhard2025}.
\AP The only known subclass for which canonical minimal automata
exist, are so-called deterministic weak Büchi automata (DWBA), for this class
minimization is possible efficiently~\cite{Staiger83,Loding01}.  

\emph{History-deterministic automata} over $\w$-words are non-deterministic
automata where the non-determinism can be resolved by a strategy depending only
on the prefix of the $\w$-word read so far. This makes them suitable for use in
game-based synthesis procedures \cite{HenzingerP06},\cite{Ehl.Kha.Fully2024};
see also \cite{BokerL23} for a recent survey. 
In this work we deal with transition-based acceptance conditions as opposed to
state-based conditions. In the following, when we speak of co-B\"uchi automata, we always refer to transition-based co-B\"uchi automata.
Co-B\"uchi
history-deterministic automata are particularly attractive because: 
\begin{itemize}
\item they can be exponentially smaller than deterministic co-B\"uchi automata for the same language~\cite{Kup.Skr.Determinisation2015},
\item there exist canonical, minimal history-deterministic co-B\"uchi automata, and minimization can
be done in deterministic polynomial time~\cite{Rad.Kup.Minimization2022}.
\end{itemize}
This is remarkable because co-Büchi languages are a standard class of
$\w$-languages that appear in many contexts. They capture the persistence
properties in the temporal hierarchy of \cite{MannaP89}, and every regular
$\w$-language can be written as a finite Boolean combination of co-Büchi
languages~\cite{Thomas90}. 

The existence of such canonical automata naturally raises the question of better
understanding their structure. For instance, we can aim to describe them using
congruences, similarly to the well-known description of minimal deterministic
finite automata using the right congruence (see standard textbooks on automata
theory like \cite{HopcroftU79}). We can then evaluate whether this new
description offers any advantages. A quite obvious choice is to apply it to the
passive learning problem which heavily relies on congruences and remains
largely unsolved for $\w$-automata. Our contributions follow this plan.

{ \medskip
\noindent
\textbf{Contribution 1} \textit{We give a congruence-based description of
minimal history-deterministic co-Büchi automata using congruence relations on
pairs of finite words (these are congruences with respect to right
concatenation of letters in the second component of the pairs).} \medskip
}

So far, we have had only an algorithmic description of the minimal
history-deterministic co-Büchi automata. The method involves taking a
history-deterministic automaton and applying a minimization
algorithm~\cite{Rad.Kup.Minimization2022}. This algorithm consists of a series
of transformations, the most complex being safe determinization from
\cite{Kup.Skr.Determinisation2015}. This step requires solving a game on two
copies of the automaton, the winning region of the game determining which states and
transitions should be retained. Although this algorithmic approach provides
polynomial-time minimization, it does not provide a direct description of
the minimal automaton just in terms of the language, because it needs some automaton for the language to start from.
Such a direct description is the essence of our first contribution. 

Our second contribution uses the congruence-based description to develop a passive learning algorithm for history-deterministic co-Büchi automata. This is the first efficient learning algorithm for a class of $\w$-automata that does not become trivial when restricted to prefix-independent languages. Finite regular languages are determined by their right congruence.
This is not the case for $\w$-languages; in particular, prefix-independent
languages have only one right congruence class. We believe that developing an
efficient learning algorithm for co-B\"uchi languages shows that we have gained
a better understanding of the class and have made measurable progress in
learning theory. 

{\medskip
\noindent
\textbf{Contribution 2} \textit{We develop a passive learning algorithm for
history-deterministic co-Büchi automata that is polynomial in time and data
(as explained below).}\medskip\noindent
}

Passive learning is the task of constructing an automaton from a given set of
labeled examples, where the label indicates whether the example word belongs to
the language or not. Optimally, the algorithm solving the passive learning problem should
be efficient in both time and data. Time efficiency means that the algorithm
should be able to produce an automaton in polynomial time with respect to the
size of the sample. Being efficient in data refers to the concept of learning
every automaton in the limit. For every language $L$, there should be a
characteristic sample such that the learner produces an automaton for $L$ with
this characteristic sample as input, and it produces the same automaton for
every extension of the characteristic sample that is consistent with $L$. Being
polynomial in data means that for every language in the class, there should be a
characteristic sample of size polynomial in the size of the minimal automaton
for the language.

This problem has been studied for DFAs since the 1970s (see
\cite{BiermannF72,TrakhtenbrotB73,Gold78}). The first passive learning algorithm
polynomial in time and in data was proposed in~\cite{Gold78}. Since then, many
variations of the basic algorithms have been developed (see \cite{LopezG16} for
a survey) and were implemented in recent years, e.g., in the library flexfringe
\cite{Ver.Ham.Flexfringe2017}.

For $\w$-automata the progress has been much slower. There is a polynomial
active learning algorithm for learning deterministic weak Büchi automata (DWBA)
\cite{MalerP95}, which allows to construct a passive learner that is polynomial
in time and data for DWBA. The minimal automata for this class have a simple
congruence-based description analogous to DFAs \cite{Staiger83}. There are some
passive learning algorithms for $\w$-automata
\cite{AngluinFS20,BohnL21,BohnL22,BohnL24}, but the language classes that can be
learned from polynomial data by these learners are defined by semantic
restrictions of the standard right congruence of the underlying language and
become trivial when restricted to prefix-independent languages (see also related
work in Section~\ref{sec:related-work}).

The primary difficulty in developing learning algorithms for $\w$-automata lies
in the absence of a canonical form and a manageable description of such a form,
preferably based on congruences. Given our first contribution, one could expect
that it would be quite easy to get a passive learning algorithm for
history-deterministic co-Büchi automata just by following the same recipe as for
the DFA case. However, we encounter important challenges because the
congruence-based description of the minimal automaton does not use all
equivalence classes, but only those we call pointed. So, unlike in the case of
finite words, the learning algorithm cannot just enumerate all equivalence
classes and then pick the pointed ones, because the overall number of classes
might be exponential in the number of pointed classes. And not only can the
number of classes be exponential, but there can also be equivalence classes
containing only representatives of exponential size in the size of the minimal
automaton. Fortunately, for pointed classes, there are always polynomial-sized
representatives. These observations indicate why the progress in learning
algorithms for $\w$-automata has been so slow. The procedure finding
the pointed equivalence classes is the most complex part of our algorithm.

The paper is structured as follows. We finish the introduction with a discussion
of related work. In Section~\ref{sec:prelim} we introduce basic terminology and
basic results. In Section~\ref{sec:automaton} we give the congruence-based
automaton description and prove its correctness. In Section~\ref{sec:learning}
we present the learning algorithm in two steps. First, we consider an idealized
algorithm allowing us to point out difficulties outlined above and to describe
how we solve them. The final algorithm is a refinement of every step of the
idealized algorithm. In Section~\ref{sec:conclusion} we conclude and point out
several directions for further research. 

\subsection{Related Work} \label{sec:related-work}

Clearly, the starting point of our work are results and insights
from~\cite{Rad.Kup.Minimization2022}, where the existence of canonical minimal
history-deterministic co-Büchi automata is shown, which itself builds
on~\cite{Kup.Skr.Determinisation2015}. But our proofs do not rely on any of
these prior results and thus provide a new and independent view on
history-deterministic co-Büchi automata. 

The congruences we use appear in related forms, either implicitly or explicitly, in the literature. We use the right congruence $\siml$ of a language, which is a standard tool in the theory of automata on finite and infinite words. The relation $\approxl$ introduced in Section~\ref{sec:learning} is a relation on pairs of words. It is a different representation of the \emph{syntactic family of right congruences} of a language introduced in \cite{MalerS97}.
The relation $\tbe$ that is introduced in Section~\ref{sec:automaton} and which
is the basis for the automaton definition is new in its exact definition, but it
is closely related to a family of automata used in \cite{BohnL22}, where a
passive learner for deterministic Büchi automata is developed. This family of
automata, one for each $\siml$ class, induces an equivalence on pairs that
corresponds to our relation $\tbe$ with an additional condition of
$\siml$-equivalence on the first component of pairs. 
Omitting this additional condition, is crucial for
minimality.
The notion of "pointed"
pairs and the selection of the corresponding equivalence classes is, to the best
of our knowledge, new. 

Concerning previous results on passive learning for $\w$-automata, there is a
polynomial time active learner for DWBAs \cite{MalerP95}, which induces a
passive learner that is polynomial in time and data (see
\cite[Proposition~13]{BohnL21}). The paper \cite{AngluinFS20} gives an
adaptation of Gold's passive learner for DFAs to $\w$-automata. This work
applies only to automata that have exactly one
state per $\siml$-class (referred to as IRC languages for informative right
congruence). Then the transition system can be inferred as for DFAs, and it only
remains to deal with the acceptance condition. The class of DWBA languages is
contained in the IRC class \cite{Staiger83}. The well-known RPNI algorithm
\cite{GarciaO92} that infers a DFA from examples by a state merging technique
has been adapted to deterministic $\omega$-automata in \cite{BohnL21}. The
algorithm requires only polynomial data for the IRC languages considered in
\cite{AngluinFS20}, and it can also infer automata for some languages beyond
this class, but there is no further characterization of the class of inferable
languages. The main lesson to learn from \cite{AngluinFS20,BohnL21} is that we
know how to deal with acceptance conditions in learning algorithms when the
transition system is easy to infer. There is some work beyond IRC languages,
presenting polynomial-time passive learners for deterministic Büchi automata
\cite{BohnL22} and deterministic parity automata \cite{BohnL24}. However, for
obtaining classes that can be learned from polynomial data, again a restriction
on the right congruence is required: the number of states for each $\siml$-class
needs to be at most $k$ for some fixed $k$. Restricted to prefix-independent
languages (a single $\siml$-class), this gives only finitely many languages for
fixed $k$. Therefore, none of the aforementioned learners is capable of
learning a nontrivial class of prefix-independent languages from polynomial
data. 

Finally, there is a polynomial-time active learner for the class of
deterministic parity automata \cite{MichaliszynO22}. But this algorithm uses, in
addition to membership and equivalence queries, so-called loop-index queries,
which provide some information on the structure of the target automaton and not
just on the target language.

\paragraph{Acknowledgement.} We thank the anonymous reviewers of our LMCS submission for their careful reading and their comments. 

\section{Preliminaries} \label{sec:prelim}

An alphabet $\S$ is a non-empty, finite set of letters. 
We use standard notation, $\S^*$ for the set of finite words, $\S^\w$ for the
set of infinite words.
We write $""\sqsubseteq""$ for the prefix relation on words.
In our learning algorithm we use \emph{length-lexicographic order} on finite
words, that is we first compare the lengths of words, and then compare them
lexicographically, assuming some fixed order on the alphabet.
Whenever we say that a word $v$ is \AP ""smaller"" than $w$, we refer to
length-lexicographic order.

A \emph{""co-Büchi automaton""} is a tuple $\Aa=(Q,Q_\init,\S,\Delta\incl
Q\times\S\times\set{1,2}\times Q)$ where $Q$ is a finite set of states, $Q_\init$ a
set of initial states, and $\Delta$ is a transition relation, each transition
having a rank $1$ or $2$ (we say $1$-transition or $2$-transition).
Without loss of generality we can assume that from every state there is an
outgoing transition on every letter.
A run of $\Aa$ from a state $q_0$ on an infinite word $a_0a_1\dots\in\S^\w$ is a
sequence of transitions
$q_0\aact{a_0:i_0}q_1\aact{a_1:i_1}q_2\aact{a_2:i_2}\cdots$ with $a_j \in \S$ and
$i_j \in \{1,2\}$. 
It is accepting if there are only finitely many
$1$-transitions on the run; in other words we work with the min-parity
condition.
We write $L(\Aa,q)$ for the set of words with an accepting run from $q$; and
$L(\Aa)$ for those having an accepting run from some state in $Q_\init$.
We write simply $L(q)$ when $\Aa$ is clear from the context. 
Finally $u^{-1}L(q)$ denotes for the set of words $v$ such that $uv\in L(q)$.

An automaton $\Aa$ is \AP""history-deterministic"" if Eve can resolve
nondeterminism in $\Aa$ while reading the input.
More precisely, Eve should have a winning ""strategy"" in the following game.
Eve starts by putting a token at some initial state.
Then the game proceeds in rounds: Adam chooses a letter, and Eve moves the token along a transition of her choice (with the input letter chosen by Adam).
Eve wins if the run she constructs is accepting or the infinite sequence chosen
by Adam is not in $L(\Aa)$.

We use arrow notation
to denote runs on finite words. 
For a finite word $w\in\S^*$ we write $q\act{w}q'$ when there is a run of the
automaton from $q$ to $q'$ on 
$w$ using only $2$-transitions; since the automaton is nondeterministic, there
 \AP can be more than one run. 
We write $q\bact{w}q'$ if there is a run from $q$ to $q'$ on $w$ but not $q\act{w}q'$.
Sometimes we omit the state $q'$ in the notation, so  
$q\act{w}$ means that there is $q'$ with $q\act{w}q'$, and $q\bact{w}$ means
that there is $q'$ with $q\bact{w}q'$, and there is no $q''$ with
$q\act{w}q''$. Note that $q\bact{w}$ implies that all runs from $q$ on $w$ have minimal rank $1$.
Finally, we simply write $q\aact{w}q'$ to say that there is a run without
specifying the ranks on it. 
Most often $\Aa$ will be clear form the context, so we do not specify it in our notation.

This notation can be used to define an important concept of a \AP""safe
language"" of a state: \AP $\iLs(\Aa,q)=\set{w\in\S^*: q\act{w}}$.
Note that the safe language of a state is a set of finite words, whereas $L(\Aa,q)$ is a set of infinite words.
We write $\Ls(q)$ when $\Aa$ is clear from the context.

Finally, we introduce a few standard notions for history-deterministic co-Büchi automata.
A co-Büchi automaton $\Aa$ is
\begin{itemize}
\item \AP""normalized"", if for each $2$-transition $q \act{a} q'$ there is $x
\in \S^*$ such that $q' \act{x} q$, so when restricted to $2$-transitions the
graph of the automaton consists of separated strongly connected components,
called the \AP ""safe SCCs"" of $\Aa$.
\item \AP""semantically-deterministic"", if $p\aact{a} q$ implies
$L(q)=a^{-1}L(p)$; it is well-known that a history-deterministic automaton can take only such
transitions~\cite{Kup.Skr.Determinisation2015}.
\item \AP""unsafe-saturated"", if for every $p$, letter $a$, and state $q$
such that $L(q)=a^{-1}L(p)$ we have $p\aact{a:1}q$; in the case of co-Büchi
automata, adding such $1$-transitions does not change the accepted language.
\item \AP""safe-deterministic"", if for every state $p$ and letter $a$ there
is at most one $2$-transition on $a$ from $p$, so the only non-determinism left
is in whether to choose a $1$ transition, and if so which one.
\end{itemize}

For all the above properties but "safe-determinism" it is easy to prove that
they can be ensured on a "history-deterministic" automaton $\Aa$ without
increasing the number of its states~\cite{Kup.Skr.Determinisation2015}.\footnote{More
precisely, in~\cite{Kup.Skr.Determinisation2015} automata cannot have parallel
transitions on different ranks, but one can simply transfer their result to our setting by first applying all the transformations from \cite{Kup.Skr.Determinisation2015} and then adding all $1$-transitions required for unsafe-saturation.}
\begin{lem}~\cite{Kup.Skr.Determinisation2015}\label{lem:co-Buchi-normal-sd}
  Every "history-deterministic" co-Büchi automaton can be made
  "normalized", "semantically-deterministic", and "unsafe-saturated" without
  changing the language, by respectively modifying the rank of some transitions,
  removing some states, and adding some $1$-transitions.
\end{lem}

Once an automaton is unsafe-saturated the whole complexity of a co-Büchi automaton is
hidden in 
understanding the structure of safe SCCs.
This indicates why safe-determinism is such a central property as it ensures that
inside a safe SCC the automaton is deterministic.
Making an automaton safe-deterministic is slightly more involved.
The only way we know how to do this is to use a game-theoretic approach of
Kuperberg and Skrzypczak~\cite{Kup.Skr.Determinisation2015}.
We do not use this result in our approach.

Assuming safe-determinism, the following result, essentially also
from~\cite{Kup.Skr.Determinisation2015}, shows how the notions introduced
above can be used together.

\begin{lem}\label{lem:co-Buchi-HD}
  Every "semantically-deterministic",  "unsafe-saturated", and "safe-deterministic" co-Büchi automaton is "history-deterministic".
\end{lem}
\begin{proof}
  Let $\Aa$ be such an automaton.
  We define a strategy to accept $L(\Aa)$. The argument is a generalization of the argument
  in the proof from \cite[Lemma~3]{Kup.Skr.Determinisation2015}, where a similar construction of a
  strategy is given for a specific language.
  We use some arbitrary linear ordering $p\leq q$ on the states of
  $\Aa$.

  Let us define a \emph{support}
  of a sequence $u\in\S^*$ to be a pair $(x,p')$ where $x$ is a prefix of $u$
  and $p'$ is a state such that for some   initial state $p$ there is a run
  $p\aact{x}p'\act{z}$ where $xz=u$.
  A \emph{base} $(x_u,p_u)$ of $u$ is a support of $u$ with $x_u$ the
  shortest possible, and $p_u$ the $\leq$-smallest once $x_u$ is fixed.
  Finally, the \emph{top state} $q_u$ of $u$ is the state reached by the run
  $p\aact{x_u}p_u\act{z_u}q_u$, where $x_uz_u=u$.
  State $q_u$ is determined by the base thanks to "safe-determinism". 
  
  The strategy for the automaton is to be in the top state of $u$ after reading
  $u$.
  Let us look why this strategy is feasible. 
  Suppose that we read a letter $a$ after $u$.
  If there is a $2$-transition $q_u\act{a}q'$, then the base does not change
  so $q'=q_{ua}$ is the top state of $ua$.
  If there is no rank $2$-transition on $a$, then the base changes to some
  $(x_{ua},p_{ua})$.
  By "semantic-determinism" for the top state $q_{ua}$ of $ua$ we get $L(q_{ua})=(ua)^{-1}L$. 
  Since $L(q_u)=u^{-1}L$ we have $L(q_{ua})=a^{-1}L(q_u)$, so there is a transition
  $q_u\aact{a:1} q_{ua}$ as desired, because $\Aa$ is "unsafe-saturated".

  We show that this strategy guarantees accepting every word from $L(\Aa)$.
  Take a word $w\in L(\Aa)$.
  Since $\Aa$ is co-Büchi, there is a prefix $u$ of $w$ such
  that the accepting run sees only $2$-transitions after reading $u$. Let $v \in \S^*$, $a \in \S$ such that $uva$ is a prefix of $w$. Then $x_{uv} \prefeq x_{uva} \prefeq u$, and either $x_{uv} \not= x_{uva}$, or  $q_{uv} \le q_{uva}$. This means that the base can change only finitely often, and 
  in consequence, for all sufficiently long prefixes of  $w$ the base is the same.
  As we have seen in the previous paragraph,  the run
  following our strategy passes through a $1$-transition only if the base changes.
  Hence, the run on $w$ following the strategy from the previous paragraph is accepting.
\end{proof}

\section{Canonical Automaton} \label{sec:automaton}
The objective of this section is to give a direct construction of a minimal
history-deterministic co-Büchi automaton for a given co-Büchi language $L$.
This minimal automaton is defined from equivalence classes of some congruence
relation determined by $L$ (cf.\ Definition~\ref{def:Aae}).  

For the whole section, fix a co-Büchi language $L \subseteq \S^\omega$.

\subsection{Definition of a canonical automaton $\Aae$}
For $u,v\in\S^*$ we write \AP $u\isiml v$ to denote the standard right congruence: 
\begin{align*}
  u\siml v&\qquad \text{if for all $w\in \S^\w$:\ \ $uw\in L$ iff $vw\in L$.}
\end{align*}
Note that $u\siml v$ iff $u^{-1}L = v^{-1}L$.

We work with pairs $(u,v)$ and for such a pair we are interested in all
ultimately periodic words of the form $u(vx)^\omega$ such that $u \siml uvx$.
Since there are only finitely many $\siml$-classes, each ultimately periodic
word can be written so that its periodic part loops on the class
reached by the prefix. The restriction to such decompositions of ultimately
periodic words is also used in the theory of families of right congruences
\cite{MalerS97} and families of DFAs \cite{Klarlund94}.

Using $L$, we define an equivalence relation on pairs of words and take some of its
equivalence classes as states of a "history-deterministic" co-Büchi automaton for $L$. 
The transitions of this automaton extend the second component of the pair:
on a $b$-transition, $b$ is appended to the second component. 
Such a transition is a $1$-transition if the resulting pair is unproductive,
and a $2$-transition otherwise.
Being unproductive intuitively means to have no extension in the language. 
We capture this by defining for every $(u,v) \in \S^* \times
\S^*$:\AP
\[
		(u,v)\iapproxbot \qquad \text{ if $v \not= \e$ and $u(vx)^\w\not\in L$ for all $x$ with $uvx \siml u$.}
\]
So we are interested in pairs $(u,v)$ such that $(u,v) \napproxbot$, and we have
a $2$-transition on $a$ from $(u,v)$ to $(u,va)$, unless $(u,va)\approxbot$.
Of course, we cannot take all $(u,v)$ as states, but rather need to define some
equivalence relation of a finite index on these pairs.
As we will see shortly, we cannot even take all the equivalence classes of the relation we
 define.

Each pair $(u,v)\in \S^*\times\S^*$ naturally defines the set of all extensions
of the second component that do not lead to $\bot$:
\[
\AP\isfl(u,v):=\set{x\in\S^*: (u,vx)\napproxbot}
\]
Intuitively, this corresponds to the concept of a safe language $\Ls$ (cf.\ Proposition~\ref{prop:Aae-accepts-L}).
The equivalence relation $\itbe$ merges $\siml$-equivalent pairs with same "safe language":\AP
\[
\text{$(u,v)\tbe(u',v')$ if $uv\siml  u'v'$ and $\sfl(u,v)=\sfl(u',v')$}
\]
We write $\ecls{u,v}$ for the $\tbe$ equivalence class of $(u,v)$.
\begin{exa} \label{ex:ak}
Consider the alphabet $\S_k := \{a_1, \ldots, a_k\}$ and the language of all
$\omega$-words that do not contain all letters infinitely often. Then $\siml$ is
trivial (all words are $\siml$-equivalent), and $(u,v) \approxbot$ iff $v$
contains all letters. The language $\sfl(u,v)$ consists of all $x$ such that $vx$ does not
contain all letters, and $(u,v) \tbe (u',v')$ if the same letters occur in $v$
and $v'$. The number of classes of $\tbe$ is thus $2^k$. The canonical
history-deterministic co-Büchi automaton for $L$ has $k$ states, one for each
$a_i$, that has self-loops with $2$-transitions on all $a_j$ with $j\not=i$, and 
all possible $1$-transitions.
\end{exa}
The above example shows that, in general, $\tbe$ has too many equivalence classes to construct the minimal automaton from them. 
For our automaton construction, we use only $\tbe$-classes of pairs that we call "pointed".
Intuitively, $(u,v)$ is "pointed" if no pair that is $\siml$-similar to $(u,v)$ and whose looping part ends with $v$ has a smaller 
safe language than $(u,v)$.

\begin{defi} \label{def:pointed}
	We say that $(u,v)$ is ""pointed"" if $(u,v)\napproxbot$,
	and
  \[
		\text{for all $u_1,u_2\in\S^*$, if $u_1u_2\siml u$ and $(u_1,u_2v)\napproxbot$ then $\sfl(u,v) = \sfl(u_1,u_2v)$}\ .
	\]
  Note that $u_1u_2\siml u$ and $\sfl(u,v) = \sfl(u_1,u_2v)$ imply $(u,v) \tbe (u_1,u_2v)$. We call a class of $\tbe$ pointed if it contains a pointed pair.
\end{defi}
While the notion "pointed" is a central definition of the paper, it is not easy to motivate it
at this stage. We try to give some intuition in the following example, and give some motivations later when we show technical properties of pointed elements and classes. 

\begin{exa} \label{ex:ak-2}
Continuing Example~\ref{ex:ak}, the "pointed" pairs are those of the form
$(u,v)$ where $v$ contains all but one of the letters. Indeed, if $v$ does not
contain $a_i$ and $a_j$ for $i \not= j$, then let $v_i$ be such that $v_iv$
contains all letters except $a_j$. Then $(u,v_iv) \napproxbot$ and $a_j \in
\sfl(u,v)$ while $a_j \not\in \sfl(u,v_iv)$ because $(u,v_iva_j) \approxbot$. So
out of the $2^k$ classes there are only $k$ classes that can contain "pointed"
elements. 
Indeed, if $(u,v)$ is such that $v$ contains all but one letter then it is
pointed, as for every $(u_1,u_2)$, either $u_2v$ contains the same letters as $v$ or
contains all the letters, and in the latter case $(u_1,u_2v)\approxbot$.
The rough intuition is as follows. There can be several reasons for a word to be in a language $L$: in the example, each $a_i$ not appearing infinitely often is such a reason. A "history-deterministic" automaton needs a component for each individual reason to check whether it is satisfied. The classes of $\tbe$ roughly correspond to subsets of reasons, and the classes of "pointed" pairs correspond to exactly one reason.
\end{exa}

Here is a more complicated example when $\siml$ is not trivial.
\begin{figure}
\scriptsize
\begin{tikzpicture}[
    class/.style={circle, draw, minimum size=.1cm, inner sep=.5mm},
    edge label/.style={draw=none, fill=none}
]
\node[class,fill=orange!50] (e) at (0, 0) {$\varepsilon$};
\node[class,fill=yellow!50] (a) at (1, 0) {$a$};
\node[class,fill=pink!50] (b) at (0, -1) {$b$};
\node[class,fill=green!50] (ab) at (1, -1) {$ab$};

\draw[->] (e) -- node[edge label, above] {$a$} (a);
\draw[->] (e) -- node[edge label, left] {$b,c$} (b);
\draw[->] 
    (a) edge[loop right] node[edge label, above] {$a,c$} ();
\draw[->]
    (b) edge[loop left] node[edge label, below] {$a,b,c$} ();
\draw[<->]
    (ab) edge node[edge label, left]{ $b$ }(a);
\draw[->]
    (ab) edge[loop right] node [edge label,right]{ $a,c$ } ();

\node[draw=none] at (-.5, 1) {$\siml:$};

\end{tikzpicture}
\hspace{3cm}
\begin{tikzpicture}[
    pointed/.style={rectangle, draw, thick, minimum size=.5cm,inner sep=.5mm},
    unpointed/.style={rectangle, draw, dashed, minimum size=.5cm,inner sep=.5mm},
    edge label/.style={draw=none, fill=none}
]

\node[pointed,fill=orange!50] (e-e) at (0, 0) {$(\varepsilon,\varepsilon)$};
\node[pointed,fill=pink!50] (b-e) at (0, -1) {$(b,\varepsilon)$};
\node[pointed,fill=yellow!50] (a-e) at (4, 0) {$(a,\varepsilon)$};
\node[unpointed,fill=green!50] (ab-e) at (1.5, -1) {$(ab,\varepsilon)$};
\node[pointed,fill=green!50] (a-b) at (4, -1) {$(a,b)$};
\node[pointed,fill=green!50] (ab-c) at (3, -1.6) {$(ab,c)$};

\draw[->] 
    (a-b) edge[loop right] node[edge label] {\scriptsize $a:2$} ();
\draw[<->]
    (a-b) edge node[edge label, right]{\scriptsize $b:2$ }(a-e);
\draw[->] 
    (a-e) edge[loop right] node[edge label] {\scriptsize $a:2$} ();
\draw[->]
    (ab-e) edge[loop below] node [edge label]{\scriptsize $a:2$ } ();


\draw[->]
    (ab-e) edge node[edge label, above]{\scriptsize $c:2$ }(ab-c);
\draw[->]
    (ab-e) edge node[edge label, above]{\scriptsize $b:2$ }(a-e);

\draw[->] 
    (ab-c) edge[loop right] node[edge label] {\scriptsize $a,c:2$} ();

\node[draw=none] at (-.5, .5) {$\tbe:$};

\end{tikzpicture}

\caption{Relations $\siml$ and $\tbe$ for the language of all words over $\{a,b,c\}$ that start with $a$ and have finitely many $c$ or a finite and odd number of $b$, see Examples~\ref{ex:a-start}, \ref{ex:a-start-aut}.} \label{fig:a-start}
\end{figure}
\begin{exa} \label{ex:a-start}
Figure~\ref{fig:a-start} shows the relations $\siml$ and $\tbe$ for the
language $L$ over the alphabet $\{a,b,c\}$ that contains all words that start
with $a$ and that have finite number of $c$'s or a finite and odd number of $b$'s. The
right congruence $\siml$ of $L$ has four classes:  the initial class
contains only $\e$, the class of $b$ contains all non-empty words that do not
start with $a$, the class of $a$ the words starting with $a$ and having even
number of $b$'s, finally the class of $ab$ the words starting with $a$ and with
an odd number of $b$'s.
For $\tbe$ relation, the $6$ classes that are $\napproxbot$ are shown. 
The color indicates the $\siml$-class of the respective $\tbe$-class, namely, 
the $\siml$-class of the concatenation of the two words of the pair.
By definition, pairs with different $\siml$-class are not $\tbe$-equivalent. Let us explain why the three pairs of $\siml$-class $ab$ in Figure~\ref{fig:a-start} are non-equivalent.
First, $(ab,\e)$ is not equivalent to $(a,b)$ because $c\in\sfl(ab,\e)$
since $abc^\omega$ is in the language, but $c\not\in\sfl(a,b)$ since $a(bcx)^\w$ is not in the language for every word $x$ because it has infinitely many
$b$ and $c$. Second,  $(ab,c)$ is not equivalent to $(ab,\e)$ or $(a,b)$ because $b \in \sfl(ab,\e)$, $b \in \sfl(a,b)$ but $b \notin \sfl(ab,c)$ using a similar reasoning as in the first case.
It is not hard to see that all other pairs are equivalent to one shown in the diagram or to $\bot$.
The
transitions in the diagram correspond to prolonging the second element of the
pair with the letter on the
transition.
For example, we have a transition $(a,b) \act{b} (a,bb)\tbe (a,\e)$. Missing $2$-transitions mean that the resulting pair is $\bot$.
For example, there is no $2$-transition on $c$ from $(a,\e)$ because $(a,c) \approxbot$, since there is no extension $x$ such that $a(cx)^\omega$ contains a finite and odd number of $b$.
The classes containing "pointed" pairs have
a solid frame (in this example, each class consists only of non-pointed pairs or
only of "pointed" pairs, but in general this needs not to be the case). The only
class that does not contain a "pointed" pair is the class of $(ab,\e)$, which
contains all pairs of the form $(u,a^n)$ for for $u \siml ab$ and $n \ge 0$.  Since $b \in \sfl(u,a^n)$
but $b \not\in \sfl(u,ca^n)$, we get that $(u,a^n)$ is not "pointed" because
$uc \siml ab$. We explain why $(a,b)$ is pointed, the reasoning for the other pairs being similar. The "safe language" of $(a,b)$ consists of all words that do not contain $c$, and a pair $(u,v)$ is $\tbe$-equivalent to $(a,b)$ iff $uv$ contains an odd number of $b$ and $v$ contains no $c$ and at least one $b$.
For any $u_1,u_2$ with $u_1u_2 \sim a$ and $(u_1,u_2b) \napproxbot$, we have that $u_1u_2$ contains an even number of $b$ (because $u_1u_2 \siml a$) and $u_2$ does not contain $c$ (because $(u_1,u_2b) \napproxbot$). Hence $(u_1,u_2b) \tbe (a,b)$, and thus $(a,b)$ is pointed.
\end{exa}

The next lemmas ensure that we can define the $2$-transitions on the level of
classes of "pointed" pairs.
\begin{lem}
  \label{lem:equiv-right}
	If $(u,v)\tbe (u',v')$ then for every letter $a$ we have $(u,va)\tbe (u',v'a)$.
\end{lem}
\begin{proof}
	Clearly $uva\siml u'v'a$, as $uv\siml u'v'$. 
	Take some $x\in\sfl(u,va)$. Then $ax \in \sfl(u,v)$, and hence $ax \in \sfl(u',v')$ because $(u,v)\tbe (u',v')$. So $x\in\sfl(u',v'a)$. The symmetric argument yields $\sfl(u,va) = \sfl(u',v'a)$.
\end{proof}

\begin{lem}\label{lem:pointed-right}
  If $(u,v)$ is "pointed" then for every letter $a$, either $(u,va)\approxbot$ or
  $(u,va)$ is "pointed".
\end{lem}
\begin{proof}
  Suppose $(u,va)\napproxbot$, we show that $(u,va)$ is "pointed".
  Take $u_1,u_2$ such that $u_1u_2\siml u$ and $(u_1,u_2va)\napproxbot$.
  We need to show that $\sfl(u,va)=\sfl(u_1,u_2va)$. 
  Since $(u_1,u_2v)\napproxbot$, from $(u,v)$ being "pointed" it follows
  that $\sfl(u,v)=\sfl(u_1,u_2v)$.
  Hence, $\sfl(u,va)=\sfl(u_1,u_2va)$ from the definition of $\sfl$.
\end{proof}

We are ready to define a canonical automaton
for $L$.

\begin{defi}\label{def:Aae}\AP
	The canonical automaton $\iAae$ is defined by:
	\begin{itemize}
		\item $Q^L=\set{\ecls{u,v} : \text{$(u,v)$ is "pointed"}}$,
		\item $Q^L_\init=\set{\ecls{u,v} : \text{$(u,v)$ "pointed" $uv\siml\e$}}$, 
		\item $\ecls{u,v}\aact{a:2}\ecls{u,va}$ if $(u,va)\napproxbot$,
		\item $\ecls{u,v}\aact{a:1}\ecls{u',v'}$ if $uva\siml u'v'$.
	\end{itemize}
\end{defi}
Note that the transitions are well-defined: for the first case by Lemma~\ref{lem:equiv-right}, and for the second case because $(u_1,v_1) \tbe (u_2,v_2)$ implies that $u_1v_1 \siml u_2v_2$, and $\siml$ is a right congruence. 

\begin{exa} \label{ex:ak-3}
For the language $L$ from Examples~\ref{ex:ak} and \ref{ex:ak-2}, the canonical
automaton $\Aae$ has one state for each letter $a_i$, which corresponds to the
class of "pointed" pairs $(\e,v_i)$, where $v_i$ contains all letters except
$a_i$, cf.~Figure~\ref{fig:ak-3-schematic}.
There are loops of $2$-transitions on this state for all letters except $a_i$.
Since $\siml$ has only one class, all states are connected by $1$-transitions
for all letters (so the $1$-transitions on every letter form a complete graph).

\begin{figure}[h]
\centering
\begin{tikzpicture}[
  state/.style={circle, draw, minimum size=.55cm, inner sep=.5mm},
  edge label/.style={draw=none, fill=none}
]
\node[state] (a1) at (0,0) {$a_1$};
\node[state] (a2) at (2.2,0) {$a_2$};
\node[draw=none] (dots) at (4.2,0) {$\cdots$};
\node[state] (an) at (6.2,0) {$a_n$};

\draw[->] (a1) edge[loop above] node[edge label, above] {$\Sigma\setminus\set{a_1}:2$} ();
\draw[->] (a2) edge[loop above] node[edge label, above] {$\Sigma\setminus\set{a_2}:2$} ();
\draw[->] (an) edge[loop above] node[edge label, above] {$\Sigma\setminus\set{a_n}:2$} ();
\end{tikzpicture}
\caption{Automaton $\Aae$ for the language $L$ from Example~\ref{ex:ak}. Regarding rank $2$
transitions each state $a_i$ has only a self-loop on all letters but $a_i$.
There are rank $1$ transitions between every pair of states on every letter.
These are not shown for readability.}
\label{fig:ak-3-schematic}
\end{figure}
\end{exa}

\begin{exa} \label{ex:a-start-aut}
  For the language $L$ from Example~\ref{ex:a-start}, the canonical automaton
  $\Aae$ is obtained by taking the $\tbe$-classes with solid frame shown in
  Figure~\ref{fig:a-start}, and the $2$-transitions connecting
  them. The $1$-transitions are inherited from the transitions shown for $\siml$ in
  Figure~\ref{fig:a-start}, namely,
  a transition on letter $a$ connecting two $\siml$-classes induces
  $1$-transitions connecting all states corresponding to the first class to all
  the states corresponding to the second class.
  For example, the transitions $(a,\e) \aact{b:1} (a,b)$ and $(a,\e) \aact{b:1}
  (ab,c)$ are induced by a transition from the yellow $\siml$-class to the green
  $\siml$-class.
  From the initial state $(\e,\e)$ we have $\aact{a:1}$ transition to $(a,\e)$.
\end{exa}

\begin{thm}\label{the:Aae}
  The automaton $\Aae$ is the minimal history-deterministic co-Büchi automaton for $L$.
\end{thm}
The rest of this section is devoted to the proof of this theorem.

\subsection{Correctness of $\Aae$}
First, we need to verify that indeed $\Aae$ accepts $L$, cf.\
Proposition~\ref{prop:Aae-accepts-L}. At some moment in the argument we need to
use the fact that $L$ is accepted by a finite history deterministic automaton.
We do this straight away by giving a variant of the definition of $(u,v)
\approxbot$ based on a normalized deterministic co-B\"uchi automaton for $L$.
Recall that a deterministic automaton has a unique initial state and for
every state and letter there is exactly one outgoing transition from this state
on this letter; so such an automaton has a single run on a given input from the
initial state. The definition below uses a concrete such automaton $\Dd$, but
Lemma~\ref{lem:approx-d} then shows that it is independent of the choice of
$\Dd$.

\begin{defi}\label{def:approx-d}
  Let $\Dd$ be a "normalized" deterministic co-B\"uchi automaton $\Dd$ for $L$.
  For $(u,v) \in \Sigma^* \times \Sigma^*$, write $(u,v) \iapproxd \bot$ if for
  all $u' \in \Sigma^*$ with  $u' \siml u$, the run of $\Dd$ on $u'v$ takes a
  $1$-transition on the suffix $v$.
\end{defi}
Observe that the definition implies $(u,\e)\napproxd \bot$, for every $u$,
because there is no run taking a $1$-transition on $\e$.

\begin{lem}\label{lem:approx-d}
  $(u,v) \approxbot$ iff $(u,v) \approxd \bot$.
\end{lem}
\begin{proof}
  Assume $(u,v) \napproxbot$.
  If $v=\e$ then $(u,v)\napproxd \bot$.
  Otherwise, let $x \in \Sigma^*$ with $uvx \siml u$ and $u(vx)^\omega \in L$.
  Then $\Dd$ accepts $u(vx)^\omega$ meaning there is a run not visiting
  $1$-transitions after some prefix $u(vx)^i$. Choosing $u' = u(vx)^i$ we get
  $u'\siml u$ from $uvx \siml u$, and the run of $\Dd$ on $u'v$ does not take a $1$-transition. 
  So, $(u,v) \napproxd \bot$.

  For the other direction, assume $(u,v) \napproxd \bot$.
  If $v=\e$ then $(u,v)\napproxbot$ by definition. 
  Otherwise, let $u'$ be such that $u' \siml u$, and the run of $\Dd$ on $u'v$
  does not visit a $1$-transition on the suffix $v$. 
  So in $\Dd$ there is a run $\aact{u'}q'\act{v} q$ for some $q,q'$. 
  Since $\Dd$ is "normalized", there is $x$ with $q \act{x} q'$ in $\Dd$. 
  We get $u'(vx)^\omega \in L$ and hence $u(vx)^\omega \in L$ because $u' \siml
  u$.  
  Moreover, $u'vx\siml u'$ because they reach the same state in $\Dd$.
  Hence, $uvx\siml u$ too (as $\siml$ is a congruence) giving $(u,v)\napproxbot$.
\end{proof}

\begin{lem}\label{lem:equiv-finite}
  Relation $\tbe$ has finitely many equivalence classes.
\end{lem}
\begin{proof}
By Lemma~\ref{lem:approx-d}, the language $\sfl(u,v)$ only depends on the set of
states that are reachable in $\Dd$ via $u' \siml u$ and the states reachable from
them by the run on $v$ not using a $1$-transition.
\end{proof}

The following observation is an immediate consequence of the definition of $\Aae$.
\begin{lem} \label{lem:Aea-sim-class}
For all "pointed" pairs $(u,v)$ and $(u',v')$ and all $x \in \Sigma^+$ we have $\ecls{u,v} \aact{x} \ecls{u',v'}$ iff $uvx \siml u'v'$.
\end{lem}
\begin{proof}
  The transitions always respect the $\siml$-class, so $\ecls{u,v} \aact{x} \ecls{u',v'}$ implies $uvx \siml u'v'$.
  For each state and each input letter, there are all possible $1$-transitions that respect the $\siml$-class. So there is a run from $\ecls{u,v}$ on $x$ that
  takes a $1$-transition to $\ecls{u',v'}$ on the last letter of $x$.
\end{proof}

The next lemma is an important technical lemma showing that for every pair
$(u,v)\napproxbot$ there is a "pointed" pair with $v$ as suffix starting in the $\siml$-class of $u$.
\begin{lem}\label{lem:make-pointed} 
	For all $u,v \in \S^*$ with $(u,v) \napproxbot$, there are $u_1,u_2 \in \Sigma^*$ with $u_1u_2 \siml u$ and $(u_1,u_2v)$  "pointed". 
\end{lem}
\begin{proof}
  We first show the following auxiliary claim:
  \begin{equation}
    \sfl(u_1,u_2v) \subseteq \sfl(u,v) \text{ for all $u_1,u_2 \in \Sigma^*$ with $u_1u_2 \siml u$.}\label{eq:after-decrease}
  \end{equation}
To see this, let $x \in \sfl(u_1,u_2v)$, that is, $(u_1,u_2vx) \napproxbot$. If $u_2vx = \e$, then $x=\e$ and $\e \in \sfl(u,v)$ because $(u,v) \napproxbot$. Otherwise, there is $y$ such that $u_1(u_2vxy)^\omega \in L$. Then also $u_1u_2(vxyu_2)^\omega \in L$ and hence $u(vxyu_2)^\omega \in L$ because $u_1u_2 \siml u$. This implies that $x \in \sfl(u,v)$, and thus shows (\ref{eq:after-decrease}).

Going back to the proof of the lemma.
If $(u,v)$ is already "pointed", then the claim holds by choosing $u_1=u$ and $u_2=\e$.

Otherwise,
by definition of "pointed", we have $u_{1,1},u_{1,2}$ with $u_{1,1}u_{1,2} \siml
u$, $(u_{1,1},u_{1,2}v) \napproxbot$, and $\sfl(u_{1,1},u_{1,2}v) \not=
\sfl(u,v)$. By (\ref{eq:after-decrease}), this means that
$\sfl(u_{1,1},u_{1,2}v) \subsetneq \sfl(u,v)$. Then $(u_{1,1},u_{1,2}v)
\ntbe (u,v)$ by the definition of $\tbe$.  
 
If $(u_{1,1},u_{1,2}v)$ is not "pointed", then we can repeat this argument,
obtaining $u_{2,1},u_{2,2}$ with $u_{2,1}u_{2,2} \siml u_{1,1}$,
$(u_{2,1},u_{2,2}u_{1,2}v) \napproxbot$ and
$\sfl(u_{2,1},u_{2,2}u_{1,2}v) \subsetneq \sfl(u_{1,1},u_{1,2}v)$. Since
$\tbe$ has only finitely many classes (Lemma~\ref{lem:equiv-finite}), this construction must terminate after
$i$ steps for some $i$, yielding the desired pair $u_1 := u_{i,1}$ and
$u_2:=u_{i,2}\cdots u_{1,2}$. A simple induction shows that $u_1u_{i,2} \cdots u_{j,2} \sim_L u_{1,j-1}$ for $j \in \{i, \ldots, 2\}$ using $u_{1,j}u_{2,j} \siml u_{1,j-1}$ and that $\siml$ is a right congruence. Hence $u_1u_2 = \underbrace{u_1 u_{i,2}\cdots u_{2,2}}_{\siml u_{1,1}}u_{1,2} \siml u_{1,1}u_{1,2} \siml u$.
\end{proof}

The following proposition more precisely characterizes the languages and the safe
languages accepted from the states of $\Aae$. 
It also justifies our notation $\sfl(u,v)$.

\begin{prop}\label{prop:Aae-accepts-L}
  For every state $\ecls{u,v}$ of $\Aae$:
  \begin{itemize}
    \item $L(\ecls{u,v}) = (uv)^{-1}L$,
    \item $\Ls(\ecls{u,v})=\sfl(u,v)$.
  \end{itemize}
  In particular, $\Aae$ accepts $L$.
\end{prop}
The proof of this proposition is split into two lemmas. 
We first show that $\Aae$ indeed accepts $L$ (Lemma~\ref{lem:Aae-accepts-L}), and then give the more precise characterizations of the languages and the safe
languages accepted from the states of $\Aae$ (Lemma~\ref{lem:residuals-in-Aae}).

\begin{lem}\label{lem:Aae-accepts-L}
  $\Aae$ accepts $L$.
\end{lem}
\begin{proof}
  Let $\Dd$ be a "normalized" deterministic co-B\"uchi automaton for $L$, as in Definition~\ref{def:approx-d}.
  
  We first show $L(\Aae) \subseteq L$: Let $w \in L(\Aae)$. Then there is a
  "pointed" pair $(u,v)$ and an index $i$ such that $\ecls{u,v}$ is reachable
  via $w[0,i]$, and $(u,vw[i+1,k]) \napproxbot$ for all $k > i$ (here,
  $w[\cdot,\cdot]$ denotes the infix of $w$ between given positions). By
  Lemma~\ref{lem:approx-d}, we have $(u,vw[i+1,k]) \napproxd \bot$ for all $k > i$.
  This means that $uvw[i+1,\infty)$ is accepted by $\Dd$. 
  Indeed, by definition of $(u,vw[i+1,k]) \napproxd \bot$  and the fact that $\Dd$ has only finitely many states,
  there must be some $q'$ reachable by some $u'\siml u$ such that the runs on $vw[i+1,k]$ from $q'$ do not see a $1$-transition.
  Hence, $vw[i+1,\infty]$ is accepted from $q'$, as $\Dd$ is deterministic.
  In consequence, $uvw[i+1,\infty)$ is in $L$.  
  Since $\ecls{u,v}$ is reachable via $w[0,i]$ from an initial state $\ecls{u',v'}$, we have $uv \siml w[0,i]$ by Lemma~\ref{lem:Aea-sim-class} and $u'v' \siml \e$,
  and thus $w = w[0,i]w[i+1,\infty) \in L$.

  For the inclusion $L \subseteq L(\Aae)$ consider an ultimately periodic word $uv^\w\in L$. 
  Without loss of generality we can assume $u\siml uv$; if not we just
  prolong $u$ and $v$.
  We show that $uv^\w$ is accepted by $L(\Aae)$.
  Let $n$ be bigger than the size of $\Aae$. Clearly $(u,v^n) \napproxbot$ as $uv^\omega \in L$ and $uv^n \siml u$. Thus, by Lemma~\ref{lem:make-pointed} there is "pointed" $(u_1,u_2v^n)$ with $u_1u_2\siml u$. 
  So there is $x \in \S^*$ with $u_1u_2v^nx \siml u_1$ such that $u_1(u_2v^nx)^\w\in L$.
  Since $u_1(u_2v^nx)^\w\in L$, we get that $(u_1,u_2v^ny) \napproxbot$ for all prefixes $y$ of $(xu_2v^n)^\omega$, and since $(u_1,u_2v^n)$ is "pointed",
  Lemma~\ref{lem:pointed-right} then implies that $(u_1,u_2v^ny)$ is "pointed"  for all prefixes $y$ of $(xu_2v^n)^\omega$. 
  Thus, by the definition of $\Aae$ we have
  $\ecls{u_1,u_2v^nxu_2}\act{v^n}\ecls{u_1,u_2v^nxu_2v^n}$. 
  Since $n$ is bigger than the size of $\Aae$, there are $i\ge0$ and $j>0$ such that 
  $\ecls{u_1,u_2v^nxu_2v^i}\act{v^j}\ecls{u_1,u_2v^nxu_2v^{i+j}}$ with
  $(u_1,u_2v^nxu_2v^i)\tbe (u_1,u_2v^nxu_2v^{i+j})$.
  We have found an accepting run on $uv^\w$ since $q_\init\aact{u}
  \ecls{u_1,u_2v^nxu_2v^i}$ by Lemma~\ref{lem:Aea-sim-class}, as
  $u_1u_2v^nxu_2v^i\siml u_1u_2v^i\siml uv^i\siml u$
\end{proof}

\begin{lem}\label{lem:residuals-in-Aae}\label{lem:safe-lang-of-state}
  For every state $\ecls{u,v}$ of $\Aae$:
  \begin{itemize}
    \item $L(\ecls{u,v}) = (uv)^{-1}L$,
    \item $\Ls(\ecls{u,v})=\sfl(u,v)$.
  \end{itemize}
\end{lem}
\begin{proof}
  Consider the first statement.
  Let $\ecls{u,v}$ be a state of $\Aae$. By Lemma~\ref{lem:Aea-sim-class},
  $\ecls{u,v}$ is reachable by $uv$, and hence $L(\ecls{u,v}) \subseteq
  (uv)^{-1}L$. For the other inclusion, let $w \in (uv)^{-1}L$. By
  Lemma~\ref{lem:Aae-accepts-L}, $uvw$ is accepted by $\Aae$. Let $\rho$ be
  an accepting run of $\Aae$ on $uvw$, and let $\ecls{u',v'}$ be the state
  in $\rho$ after $uva$, where $a$ is the first letter of $w$. Then $uva
  \siml u'v'$ by Lemma~\ref{lem:Aea-sim-class} and the definition of the initial
  states of $\Aae$. By definition of $\Aae$, there is an $a$-labeled $1$-transition from
  $\ecls{u,v}$ to $\ecls{u',v'}$, so we get an accepting run on $w$ starting in $\ecls{u,v}$.
  
  Now we proceed to the second statement.
  By definition of $\Aae$ we have that $\ecls{u,v}\act{x}\ecls{u,vx}$ implies $(u,vx) \napproxbot$.
  The latter is the definition of $x \in \sfl(u,v)$ and hence $\Ls(\ecls{u,v}) \subseteq \sfl(u,v)$.
  If $(u,vx) \napproxbot$, then $(u,vy) \napproxbot$ for all prefixes $y$ of $x$, and hence $\ecls{u,v}\act{x}\ecls{u,vx}$,
  which shows the other inclusion.
\end{proof}

This finishes the proof of Proposition~\ref{prop:Aae-accepts-L}.
Thanks to this proposition, and Lemma~\ref{lem:co-Buchi-HD} it is quite
easy to see that $\Aae$ is history-deterministic.

\begin{lem}\label{lem:Aae-hd}\label{lem:Aae-sd-and-us}\label{lem:Aae-is-normalized}
  $\Aae$ is "semantically-deterministic", "unsafe-saturated",
  "safe-deterministic", and "normalized".
  In particular, $\Aae$ is "history-deterministic".
\end{lem}
\begin{proof}
  "Semantic-determinism" is implied by Proposition~\ref{prop:Aae-accepts-L} and the
  fact that all the transitions of $\Aae$ respect the $\siml$-class. Since
  $\Aae$ contains all $1$-transitions that respect the $\siml$-class,
  Proposition~\ref{prop:Aae-accepts-L} implies that $\Aae$ is "unsafe-saturated". 
  As $\Aae$ is "safe-deterministic" by definition, Lemma~\ref{lem:co-Buchi-HD}
  implies that $\Aae$ is "history-deterministic".

  It remains to show that $\Aae$ is "normalized".
  Consider a state of $\Aae$ identified by a "pointed" pair $(u,v)$.
  Suppose $\ecls{u,v}\act{x}\ecls{u,vx}$ with $x\not=\e$.
  By definitions of transitions in $\Aae$, the pair $(u,vx)$ is "pointed".
  Hence, $(u,vx)\napproxbot$.
  This means that there is $y$ such that $uvxy\siml u$ and $u(vxy)^\omega\in L$.
  In particular, $(u,vxyv)\napproxbot$, because $u(vxyvxy)^\omega\in L$.
  Now we observe that $uvxyv\siml uv$ so since $(u,v)$ is "pointed" we get
  $(u,v)\tbe (u,vxyv)$.
  Thus, $\ecls{u,vx}\act{yv}\ecls{u,v}$.
\end{proof}

\subsection{Minimality of $\Aae$}
It remains to show that $\Aae$ is minimal.
This proof is slightly more involved. Similarly to~\cite{Rad.Kup.Minimization2022}
we show that there is an injection from $\Aae$ to every history-deterministic
co-Büchi automaton for $L$ (Lemma~\ref{lem:Aae-min}). 
As a new tool, we use of
the notion of "central sequence" (Definition~\ref{def:central} and
Lemma~\ref{lem:Aae-central}). This allows us to directly give a proof for all
history-deterministic co-Büchi automata for $L$, while
\cite{Rad.Kup.Minimization2022} relies on \cite{Kup.Skr.Determinisation2015} in
order to restrict to "safe-deterministic" ones. We start with a small lemma showing that $\Aae$ has the property called safe-minimal in \cite{Rad.Kup.Minimization2022}.\looseness=-1

\begin{lem}\label{lem:states-in-Aae}
  If $L(\ecls{u,v})=L(\ecls{u',v'})$ and $\Ls(\ecls{u,v})=\Ls(\ecls{u',v'})$
  then $\ecls{u,v}=\ecls{u',v'}$.
\end{lem}
\begin{proof}
  This is a direct consequence of Lemma~\ref{lem:residuals-in-Aae}, because first,
  $(uv)^{-1}L=L(\ecls{u,v})=L(\ecls{u',v'})=(u'v')^{-1}L$, which shows $uv\siml u'v'$, and second
  $\sfl(u,v)=\Ls(\ecls{u,v})=\Ls(\ecls{u',v'})=\sfl(u',v')$, which shows $(u,v) \tbe (u',v')$.
\end{proof}

\begin{defi} \label{def:central} 
  Let $\Bb$ be a "safe-deterministic" co-Büchi automaton. 
	A ""central sequence""	for a state $q$ of $\Bb$ is $z_q\in\S^*$ such that
	$q\act{z_q}q$ and for every $p$ with $L(p) = L(q)$ we have $p\act{z_q}q$ or
	$p\bact{z_q}$.
\end{defi}

\begin{rem}
  Recall that $p\bact{z_q}$ means that every run from $p$ on $z_q$ visits a $1$
  transition, and there is at least one such run. 
  A "central sequence" of $q$ can be $\e$. In this case the definition degenerates to
  saying that there is no state $p\not= q$ with $L(p)=L(q)$.\looseness=-1
\end{rem}

\begin{lem}\label{lem:Aae-central}
  Every state $q$  of $\Aae$ has a "central sequence" $z_q$.
\end{lem}
\begin{proof}
  A state of $\Aae$ is an equivalence class $\ecls{u,v}$ of a "pointed" pair
  $(u,v)$, in particular $(u,v) \napproxbot$. 

  When $v=\e$ then $(u,v)$ is the unique state recognizing
  $L(\ecls{u,v})$.
  Indeed, if $L(\ecls{u,v}) = L(\ecls{u',v'})$\sloppy then we have $u'v' \siml uv$. 
  This gives also $u'v'\siml u$, as $v=\e$.
  Since $(u,v)$ is "pointed" we get $\sfl(u',v')=\sfl(u',v'v) = \sfl(u,v)$.
  So $(u,v)\tbe (u',v')$. Hence $\e$ is a central sequence for $\ecls{u,v}$.
  
  Now consider the case $v \not= \e$.
  Since $(u,v) \napproxbot$, there exists $x$ such that $uvx \siml u$ and
  $u(vx)^\omega \in L$. We claim that $xv$ is "central" for $q = \ecls{u,v}$. 
  
  Because $u(vx)^\omega \in L$, we have that $(u,vxv) \napproxbot$. Thus, since $(u,v)$ is "pointed", and $uvx \siml u$, we obtain that $(u,vxv) \tbe (u,v)$, which means that $\ecls{u,v} \act{xv} \ecls{u,v}$. So the first condition in the definition of "central sequence" is satisfied.

  It remains to show the second condition of the definition of "central sequence".
  Let $\ecls{u',v'}$ be some state of $\Aae$ that is language equivalent to 
  $\ecls{u,v}$. By Lemma~\ref{lem:residuals-in-Aae}, this implies that $u'v'
  \siml uv$ and thus $u'v'x \siml u$. Hence, again because $(u,v)$ is "pointed",
  we have that either $(u',v'xv) \approxbot$ or $(u',v'xv) \tbe (u,v)$. This
  shows that $xv$ is "central" for $q = \ecls{u,v}$.
\end{proof}

The following lemma when put together with Lemma~\ref{lem:states-in-Aae} shows
that $\Aae$ is a minimal history-deterministic co-B\"uchi automaton for $L$. 
\begin{lem}\label{lem:Aae-min}
  Let $\Bb$ be a "history-deterministic" co-B\"uchi automaton for $L$. For each state $p$ of $\Aae$ there is a state $q_p$ of $\Bb$ with $\Ls(\Aae,p) = \Ls(\Bb,q_p)$ and $L(\Aae,p) = L(\Bb,q_p)$.
\end{lem}
\begin{proof}
  We can assume that $\Bb$ is "normalized" and "semantically-deterministic" by Lemma~\ref{lem:co-Buchi-normal-sd}.

  Let $p$ be a state of $\Aae$. By Lemma~\ref{lem:Aae-central}, there is a
  "central sequence" $z_p$ for $p$. 
  Call a state $q$ of $\Bb$ a \emph{$z_p$-loop} if $L(\Bb,q) = L(\Aae,p)$
  and $q \act{z_p^k} q$ for some $k \ge 1$.

  We first show that $\Ls(\Bb,q) \subseteq \Ls(\Aae,p)$ for every $q$ that is a
  $z_p$-loop. 
  Suppose for a contradiction that $x \in \Ls(\Bb,q) \setminus \Ls(\Aae,p)$. Then $x \not= \e$ because $\e$ is in $\Ls$ of every state. Let $y \in \Sigma^*$ with $q \act{xy} q$. Such $y$ exists because $\Bb$ is "normalized".
  Then $\Bb$ accepts $(z_p^kxy)^\omega$ from $q$.  But since $z_p$ is "central"
  for $p$ and $x \not\in \Ls(\Aae,p)$, we have that when starting from a
  state $p'$ with $L(\Aae,p) = L(\Aae,p')$, $z_p^kxy$ must see a $1$-transition.
  By semantic determinism of $\Aae$ and the fact that $q \act{z_p^kxy} q$, we get that every run $\Aae: p \aact{(z_p^kxy)^m} p'$ ends in a state with $L(\Aae,p) = L(\Aae,p')$.
  So $(z_p^kxy)^\omega$ is
  not accepted from $p$, contradicting $L(\Bb,q) = L(\Aae,p)$.

  Now assume that for every $z_p$-loop $q$ of $\Bb$ there is $x_q \in
  \Ls(\Aae,p) \setminus \Ls(\Bb,q)$.
  Let $y_q$ be such that $p \act{x_qy_q} p$.
  Let $f$ be a "strategy" witnessing "history-determinism" of $\Bb$, and
  let $u \in \Sigma^*$ such that $p$ is reachable via $u$ in $\Aae$.
  Note that all the states of $\Aae$ are reachable because all $\siml$-classes are reachable from $\e$, and $\Aae$ has all $1$-transitions respecting the $\siml$-class.
  Consider the run of $\Bb$ constructed by $f$ of the following form:
  \[
 (\Bb,f):  q_0 \aact{uz_p^{n_0}} q_1 \bact{x_1y_1z_p^{n_1}} q_2 \bact{x_2y_2z_p^{n_2}} q_3 \cdots
  \]
  where each $n_i$ is chosen such that $q_{i+1}$ is a $z_p$-loop, $x_i := x_{q_i}$, and $y_i := y_{q_i}$.
  Note that a $z_p$-loop is reached with sufficiently many iterations of $z_p$ because after each $z_p^n$ in the above run, the state of $\Bb$ must have the same language as $p$ (otherwise $f$ would not be a strategy for accepting the correct language).  
By choice of the $x_i$, the constructed run visits $1$-transitions on all $x_i$ and thus is not accepting. But in $\Aae$ there is an accepting run that moves to $p$ on $u$ and then loops on $p$ without visiting $1$-transitions. This contradicts that $f$ constructs an accepting run for each word in $L$, so there must be a $z_p$ loop $q$ that has the same safe language as $p$.
\end{proof}

\section{Passive Learning} \label{sec:learning}
We develop a passive learning algorithm of history-deterministic co-Büchi
automata. 
It relies on the congruence-based construction of a minimal history-deterministic
automaton from the last section; Definition~\ref{def:Aae}.

\AP A ""sample"" $S$ is a set of pairs $(u,v) \in \Sigma^* \times \Sigma^+$
partitioned into sets $S^+$ and $S^-$; we often write $S = (S^+,S^-)$. A pair
$(u,v)$ represents the ultimately periodic word $uv^\omega$. 
An ultimately
periodic word can be represented by different pairs (e.g., $(a,b)$ and $(ab,bbb)$
represent the same ultimately periodic word).
When we say that $uv^\omega$ is in the sample, we mean that some pair representing $uv^\omega$ is in the sample.
\AP\ A "sample" $S$ is ""consistent"" with $L \subseteq \S^\omega$ if 
$S^+ \subseteq L$ and $L \cap S^- = \emptyset$.
The size of a "sample" is the sum of the lengths $|uv|$ over all the $(u,v)$ in
the sample. We say that $S_2$ is an
""extension"" of $S_1$ if $S^+_1\incl S^+_2$ and $S^-_1\incl S^-_2$, and we
write $S_1 \incl S_2$ in this case; here we mean the inclusion between 
$S^+_1$ and $S^+_2$ as sets of pairs of finite words, and not between the sets of ultimately periodic words they represent.

\AP A ""passive learner"" for history deterministic co-Büchi automata (just learner, for short) is a function $f$ that maps samples to history deterministic co-Büchi automata. Such a learner
$f$ is called a polynomial-time learner if $f$ can be computed in polynomial time (as usual, measured in the size of the input, i.e., the sample), and 
$f$ is a \AP  ""consistent learner"" if for each "sample" $S=(S^+,S^-)$, the
constructed automaton is "consistent" with $S$, meaning the language of the
automaton $f(S)$ is consistent with the "sample" $S$.
Further, we say that $f$ can learn every co-Büchi language from
""polynomial data"" if for every co-Büchi language $L$ there is a \AP
""characteristic sample"" $S_L$ of size polynomial in $\Aae$, where "characteristic" means that for $\Aa := f(S_L)$ we have $L(\Aa) =
L$ and for every "extension" $S$ of $S_L$ that is consistent with $L$ we have $f(S) = \Aa$.

\subsection{Overview and Challenges} \label{sub:overview}
Assuming that the "sample" contains the relevant information for the language $L$,
our "learner" constructs the automaton $\Aae$ by inferring the classes of $\tbe$
for "pointed" pairs. 

The first step is to find representatives \AP $\iRsim$\label{def:Rsim} for the classes
of $\siml$, which is used in the definition of $\tbe$. This is straightforward
and can be done as for finite words. 
We start with $\e\in \Rsim$. 
In a loop, having already $\Rsim=\set{u_1,\dots,u_k}$, we find the
"smallest" word $u$ such that for all $i$ there is a proof in $S$ of $u\nsiml u_i$.
A proof consists of two witnesses $(uw,x)\in S^+$ and $(u_iw,x)\in S^-$ for
some $w,x$ (or with $S^+$, $S^-$ interchanged).
If the "sample" contains all the minimal representatives of
$\siml$ classes, and contains proofs that they are pairwise not $\siml$-equivalent
then this algorithm indeed finds the minimal representatives of
the $\siml$-classes. 
The algorithm terminates, as the number of candidates for $u$ is bounded
by the size of $S$.\looseness=-1

For finding representatives of the "pointed" $\tbe$-classes, there are two main
challenges. The first one is that non-equivalence for $\tbe$ can, in general,
not be witnessed by finitely many examples: Spelling out the full definition of
$\ntbe$ shows that it contains a universal quantifier over finite words, so it
is not sufficient to give finitely many examples for proving $(u,v) \ntbe
(u',v')$. The second challenge is that we need to extract the "pointed" classes.
Since the definition of "pointed" contains a universal quantifier and uses the
relation $\tbe$, it is not directly possible to prove that a pair is "pointed"
just by adding finitely many examples to the "sample". Note that constructing all
classes of $\tbe$ is not an option since  the total number of $\tbe$-classes can
be exponential in the size of the minimal automaton (see Example~\ref{ex:ak-2}).
Furthermore, there are example languages showing that $\tbe$ can have classes
whose shortest representatives have length exponential in the number of
"pointed" classes (see Example~\ref{ex:long-SCC-path}).

Our plan for solving these challenges is as follows:
\begin{itemize}
  \item In (Section~\ref{sub:approx}) we introduce relation $\approxl$, which is a refinement of $\tbe$, and we show that it can be used instead of $\tbe$ on "pointed" elements.
  
  \item We immediately put $\approxl$ to work in Section~\ref{sub:components} by showing how to construct a
  "safe SCC" of $\Aae$ for a given "pointed" $(u,v)$. This turns out to be easy, essentially it is the
  same as for languages of finite words because non-equivalence $\napproxl$ can be witnessed as easily as $\nsiml$.
  
  \item In Section~\ref{sub:idealized} we present a complete but idealized learning algorithm assuming
  that we can query some properties of $L$. This idealized version is easier to understand, and its
  correctness proof serves as an intermediate step for the correctness proof of
  the passive learning algorithm.
  
  \item The most complicated step in the idealized algorithm is to find a fresh
  "pointed" pair $(u,v)$. For this we need new theoretical developments in form of
  a few technical lemmas. This step is presented in a separate Section~\ref{sub:find-pointed}. There we also give an example illustrating the main steps of the algorithm (Example~\ref{ex:long-SCC-path-learn}).
  
  \item Finally, in Section~\ref{sub:passive-learner}, we spell out the passive learning algorithm that is obtained by essentially
  replacing each line of the idealized algorithm by a passive learning
  procedure. We prove correctness of each of these procedures.
\end{itemize}

\subsection{Equivalence $\approxl$} \label{sub:approx}
The first step in our plan is to introduce an equivalence relation $\iapproxl$
that we can use instead of $\tbe$ in some contexts.
It is central to our learning algorithm.
\begin{defi}\label{def:approx}
  For $(u,v),(u',v')\in\S^*\times\S^+$ we define $(u,v)\approxl(u',v')$ if
   $u\siml u'$, $uv\siml u'v'$, and for every $x\in\S^*$ with $uvx\siml u$ we have
   $u(vx)^\w\in L$ iff $u'(v'x)^\w\in L$. 
\end{defi}
Observe that we require that the second components of the pairs in
the relation are not empty as otherwise a statement like $u(vx)^\w\in L$ would
not make sense when $vx=\e$. 
The $\approxl$ relation is easy to work with in the context of learning because
 a proof for non-equivalence $(u,v)\napproxl(u',v')$ requires only two pairs $(u,vx),(u'v'x)$ with $u(vx)^\w \in L$ iff $u'(v'x)^\w \notin L$.
This is in contrast to the $\tbe$ relation that is defined by a nested
quantification. 
The following crucial lemma shows that once we have a "pointed" pair, we can work with
$\approxl$ instead of $\tbe$.

\begin{lem}\label{lem:pointed-approx-equiv}
  For $(u,v)$ "pointed" with $v\not=\e$, and for arbitrary $x,y\in\S^*$ we have: $(u,vx)\tbe (u,vy)$ iff
	$(u,vx)\approxl (u,vy)$.
\end{lem}
\begin{proof}
  Consider the right-to-left direction. 
  We have $uvx\siml uvy$. Take $z\in\sfl(u,vx)$. 
  This means there is $z'$ such that $uvxzz'\siml u$ and $u(vxzz')^\w\in L$.
  By $(u,vx)\approxl (u,vy)$ we get $u(vyzz')^\w\in L$, so, as $uvyzz'\siml uvxzz'\siml u$, we have $z\in\sfl(u,vy)$, and we are done.

  For the left-to-right direction, we once again have $uvx\siml uvy$.
  Take $z$ such that $u(vxz)^\w\in L$ and $uvxz\siml u$.
  Since $(u,v)$ is "pointed" and by Lemma~\ref{lem:Aea-sim-class}, we have a run in $\Aae$:
  $\ecls{u',v'}\aact{uv}\ecls{u,v}\act{x}\ecls{u,vx}\act{z}\ecls{u,vxz}\act{v}\ecls{u,vxzv}$ for some initial state $\ecls{u',v'}$.
  Now from the definition of "pointed" applied to $(u,v)$ and $uvxz\siml u$ we get
  $\sfl(u,vxzv)=\sfl(u,v)$.
  This means $\ecls{u,vxzv}=\ecls{u,v}$, so we have found a loop in $\Aae$ on
  rank $2$ transitions.
  But since $(u,vx)\tbe (u,vy)$ we get
  $\ecls{u,v}\act{y}\ecls{u,vy}\act{z}\ecls{u,vxz}\act{v}\ecls{u,vxzv}$, so 
  $u(vyz)^\w\in L$ as desired.
\end{proof}

\subsection{Components $C(u,v)$ and automaton $\Aa[\Cc]$} \label{sub:components}

Our next goal is to construct for a given "pointed" $(u,v)$ an automaton $C(u,v)$
isomorphic to the "safe SCC" of $\Aae$ containing $\ecls{u,v}$.
We call such an automaton a ""component"".
Then for a set of components $\Cc$ we can define an automaton $\AaCc$ as $\Aae$
restricted to these components.

To define $C(u,v)$ we profit from Lemma~\ref{lem:pointed-approx-equiv} allowing us to work
with $\approxl$ instead of $\tbe$.
The set of states of $C(u,v)$ consists of representatives of $\approxl$-equivalence classes 
extending $(u,v)$, namely the set $\iRappl(u,v)=\set{(u,vw_1),\dots,(u,vw_k)}$\label{def:Rappl}  of pairs such that
\begin{itemize}
  \item $w_1=\e$,
  \item $w_i$ is the "smallest" such that $(u,vw_i)\napproxl\set{\bot,(u,vw_1),\dots,(u,vw_{i-1})}$, and
  \item for every pair $(u,vw)$ extending $(u,v)$ we have  $(u,vw)\approxl(u,vw_i)$ for some $i$.
\end{itemize}
The set $\Rappl(u,v)$ can be constructed using the same principle as for the construction of $\Rsim$ described at the beginning of Section~\ref{sub:overview}.

The transitions of $C(u,v)$ are determined by:
\begin{itemize}
	\item $(u,vw_i)\act{a}(u,vw_j)$ if $(u,vw_ia)\approxl(u,vw_j)$.
\end{itemize}
So a component has only rank $2$ transitions. 

\begin{lem}\label{lem:component-is-safe}
  If $(u,v)$ is "pointed" then $C(u,v)$ is isomorphic to the
	"safe SCC" of $\Aae$ containing $\ecls{u,v}$.
\end{lem}
\begin{proof}
	Since $(u,v)$ is "pointed", all $(u,vw_i)$ are "pointed" by
	Lemma~\ref{lem:pointed-right}. 
	Due to the choice of $w_1,\dots,w_k$ and by
	Lemma~\ref{lem:pointed-approx-equiv}, there is a bijection between the states 
	of the safe SCC of $\Aae$ and states of the component. 
	The transitions are defined in the same way for the two.
\end{proof}

\begin{defi}
  Given some set of "components" $\Cc=\set{C(u_1,v_1),\dots,C(u_k,v_k)}$ we
  define 
  automaton $\iAaCc$ whose states are the states of these components, whose rank
  $2$ transitions are the transitions in these components, and whose rank $1$
  transitions are: $(u,v)\aact{a:1}(u',v')$ if $uva\siml u'v'$.
  The initial states are $(u,v)$ such that $uv\siml \e$.
\end{defi}

\begin{restatable}{lem}{ComponentAutomatonIsOk}\label{lem:component-automaton-is-ok}
	For every set of components $\Cc$, we have $L(\AaCc)\incl L(\Aae)$. 
	There is a set of components such that $L(\AaCc)=L(\Aae)$.
\end{restatable}
\begin{proof}
	Every component corresponds to a "safe SCC" of $\Aae$.
	The rank $1$ transitions in $\AaCc$ are placed exactly as in $\Aae$. 
	So the graph of $\AaCc$ is a subgraph of $\Aae$, hence $L(\AaCc)\incl
	L(\Aae)$.
	If $\Cc$ contains a component for every "safe SCC" of $\Aae$ then we get the equality.
\end{proof}

\subsection{Idealized learning algorithm} \label{sub:idealized}
We are ready to present an idealized version of the learning algorithm in
Listing~\ref{lst:algo} that assumes unrestricted access to the language $L$.
In Section~\ref{sub:passive-learner}, we discuss how to implement each of these steps in a true passive learning
algorithm that does not have access to $L$ but only to a "sample" of $L$.

\begin{lstlisting}[caption={An idealized learning algorithm constructing $\Aae$},label=lst:algo,float,frame=lines,escapechar=|]
function $\mathit{IdealizedLearn}(L)$:
  find the set $\Rsim = \{u_1,\ldots,u_k\}$ 
  find $\NTsim\incl \Rsim$ |\label{alg:lineE}|
  $\Cc:=\set{ C(u,\e) : u\in \Rsim-\NTsim}$ |\label{alg:line-u-epsilon}|
  $\Aa:=\AaCc$
  while $L\not=L(\Aa)$ do |\label{alg:while}|
    find the $\text{\kl{smallest}}$ $u_i$ such that there is $x$ with $u_ix^\w\in L-L(\Aa)$, and $u_ix\siml u_i$. |\label{alg:u}|
    find the $\text{\kl{smallest}}$ $x$ such that $u_ix^\w\in L -L(\Aa)$ and $u_ix\siml u_i$|\label{alg:x}|
    $d:=\max(\set{|C| :  C\in \Cc} \cup \{1\})$
    find a $\text{"pointed"}$ $(u_i,x^dv)$ extending $(u_i,x^d)$ |\label{alg:pointed}|
    construct the $\text{"component"}$ $C(u_i,x^dv)$ using $\Rappl(u_i,x^dv)$ |\label{alg:C}|
    add $C(u_i,x^dv)$ to $\Cc$.
    $\Aa:=\AaCc$
  return($\Aa$)
\end{lstlisting}
Let us look closer how this idealized version of the algorithm works.
The first step is  to find representatives for the classes of 
$\siml$. 
For this the algorithm computes the set $\Rsim=\set{u_1,\dots,u_k}$ as explained at the beginning of Section~\ref{sub:overview}.

Then the algorithm proceeds with constructing all the "components" of $\Aae$.
There are two types of components.
We have trivial components corresponding to $\ecls{u,\e}$ with $ux^\w\not\in L$ for
all $x\in\S^+$. Observe that such $(u,\e)$ is "pointed", and there are no rank $2$ transitions
going out of $\ecls{u,\e}$.
For the language from Examples~\ref{ex:a-start}, \ref{ex:a-start-aut} and Figure~\ref{fig:a-start}, these are the components of the classes $\ecls{(\e,\e)}$ and $\ecls{(b,\e)}$.
The algorithm computes the set $\iNTsim\incl \Rsim$ of those $u_i$
for which there is $x$ with $u_i\siml u_ix$ and $u_ix^\w\in L$. 
These classes of $\siml$ give rise to non-trivial components in $\Aae$.
Again in the example from Figure~\ref{fig:a-start}, the set $\NTsim$ is $\{a,ab\}$.
If $u_i\not\in \NTsim$ then $C(u_i,\e)$ is a trivial component consisting of one
state with no transitions: $(u_i,\e)$ is "pointed" and $\sfl(u_i,\e)=\set{\e}$.
All such components are added to $\Cc$, and at this point $\AaCc$ is the
automaton accepting the empty language as it has no rank $2$ transitions. 

In the while-loop of the algorithm we keep the invariant that $L(\AaCc)\incl L$. 
If the inclusion is strict, then we find the "smallest" example for this, namely we
find the "smallest" $\siml$-class $u_i$ and the "smallest" $x$ for this class such that
$u_ix^\w\in L-L(\Aa)$ and $u_ix\siml u_i$. 
Then we find a "pointed" pair $(u_i,x^dv)$ extending $(u_i,x^d)$, for sufficiently big
$d$.
The choice of $d$ is such that it guarantees that $(u_i,x^dv)$ is new, namely
$(u_i,x^dv)\ntbe (u',v')$ for all $(u',v')\in\bigcup\Cc$ (cf. Lemma~\ref{lem:xd-gives-new}).
Hence, the component $C(u_i,x^dv)$ that we add to $\Cc$ strictly increases the size of
$\Cc$. 
This guarantees termination in polynomial time as the number of equivalence
classes of $\tbe$ with a "pointed" pair is exactly the size of $\Aae$.
In the example from Figure~\ref{fig:a-start}, after adding the trivial components, we would get $u_i = a$ and $x=a$ and $d=1$. The pair $(a,a)$ is already pointed, so the algorithm proceeds with constructing the component containing the classes $\ecls{(a,a)} = \ecls{(a,\e)}$ and $\ecls{(a,b)}$.

All the steps but that of finding a new "pointed" pair in line~\ref{alg:pointed} are relatively
easy to implement in a passive learning algorithm.
We discuss  line~\ref{alg:pointed} in Section~\ref{sub:find-pointed}. 
At this point we can already prove correctness and the polynomial time
complexity of the algorithm (Theorem~\ref{thm:idealized}) assuming every step is
implemented correctly.
The proof is based on the next lemma ensuring that the algorithm adds a new
component in every iteration of the loop.
More precisely, it says that once we have found $(u,x^d)$, we can be sure that
all its extensions $(u,x^dv)$ are new, as long as they are "pointed".

\begin{lem}\label{lem:xd-gives-new}
  With the notations from Listing~\ref{lst:algo}, suppose $ux^\w\in L-L(\AaCc)$ and $ux\siml u$. 
  Suppose moreover that $d=\max(\set{|C|: C\in\Cc} \cup\set{1})$.
  If $(u,x^dv)$ is a "pointed" pair extending $(u,x^d)$ then 
  $(u,x^dv)\ntbe (u',v')$ for every $(u',v')\in\bigcup\Cc$.
\end{lem}
\begin{proof}
  First we observe that $x^d\not\in\sfl(u',v')$ for every
  $(u',v')\in\bigcup\Cc$ with $u\siml u'v'$.
  Indeed, otherwise $\AaCc$ would have an accepting run on $ux^\w$ by the choice of $d$. 

  Since $(u,x^dv)$ is "pointed", $(u,x^dv)\napproxbot$. 
  So there is $y$ with $u(x^dvy)^\w\in L$ and $u\siml ux^dvy$. 
  In particular, $yx^d\in \sfl(u,x^dv)$.
  Suppose to the contrary that $(u,x^dv)\tbe (u',v')$ for  some $(u',v')\in C$
  and $C\in\Cc$.
  Hence, $yx^d\in \sfl(u',v')$.
  But then, since $C$ is a component, we have some $(u'',v'')\in C$ with $(u'',v'')\tbe (u',v'y)$ in $C$.
  This means that $x^d\in\sfl(u'',v'')$. 
  Moreover,  $u''v''\siml u$ because $u''v''\siml u'v'y$, $u'v'\siml ux^dv$ and
  $ux^dvy\siml u$.
  A contradiction with our observation at the beginning of the proof.
\end{proof}

\begin{thm}\label{thm:idealized}
  The idealized algorithm from Listing~\ref{lst:algo} does at most
  $|\Aae|$ iterations of the loop and returns $\Aae$.
\end{thm}
\begin{proof}
  The main point is to show the invariant stating that $\Cc$ is a subset of the "components" of $\Aae$.
  Then by Lemma~\ref{lem:component-automaton-is-ok} we have $L(\AaCc)\incl
  L(\Aae)$.
  This invariant implies correctness.
  As Lemma~\ref{lem:xd-gives-new} shows, in every iteration of the loop we add a new component.  
  Hence, the number of iterations of the loop is bounded by the size of $\Aae$.

  Let us look why the invariant holds.
  The invariant is preserved in the loop thanks to
  Lemma~\ref{lem:component-is-safe} as each time we add $C(u_i,x^dv)$
  for a "pointed" $(u_i,x^dv)$. 
  It remains to check that the invariant holds in the beginning when  $\Cc$ contains $C(u,\e)$ for all $u\in \Rsim-\NTsim$.
  We need to verify that $(u,\e)$ is "pointed". 
  For this we show that for every $u_1u_2\siml u$ with $u_2\not=\e$ we have
  $(u_1,u_2)\approxbot$. 
  Indeed, if $(u_1,u_2) \napproxbot$, then there would be $x$ with $u_1\siml
  u_1u_2x$ and $u_1(u_2x)^\w\in L$. 
  But then $u_1u_2(xu_2)^\w \in L$, hence $u(xu_2)^\w\in L$ meaning $u\in \NTsim$.
\end{proof}

\subsection{Finding a new pointed element} \label{sub:find-pointed}
The difficult part of the idealized learning algorithm from
Listing~\ref{lst:algo} is hidden in line~\ref{alg:pointed}. 
The task is to extend (the second component of) a given pair $(u,v)\napproxbot$,
satisfying $u\siml uv$,  to a  "pointed" pair. 
An idealized algorithm is presented in Listing~\ref{lst:find-pointed}.
But to explain the idea behind it, we need a definition and some lemmas. 

\begin{lstlisting}[caption={Finding a "pointed" $(u,w\bv)$ extending $(u,v)$},label=lst:find-pointed,float,frame=lines,escapechar=|]
  function $\iFindPointedL(u,v)$: //required $u\siml uv$
    Set $\bv=v$.
    while $\exists x$ s.t. $u\siml u\bv x$, $(u,\bv x \bv)\napproxbot$ and $u(\bv x)^\w\not\in L$ |\label{fp:while1}|
      Find the $\text{\kl{smallest}}$ such $x$. |\label{fp:x1}|
      $\bv:=\bv(x\bv)^m$, where $m$ the biggest s.t., $(u,\bv(x\bv)^m)\napproxbot$. |\label{fp:m}|
    Set $w=\e$.
    while $\exists x$ s.t. $u\siml uw\bv x$, $(u,w\bv x\bv)\napproxl\set{\bot,(u,w\bv)}$ |\label{fp:while2}|
      Find the $\text{\kl{smallest}}$ such $x$. |\label{fp:x2}|
      $w:=w\bv x$.
    return $(u,w\bv)$  
\end{lstlisting}

We heavily use the structure of $\Aae$ in this subsection. 
We use $p,q$ for states of $\Aae$. 
Whenever we talk about states or transitions, this refers to the automaton
$\Aae$.
We write $\aact{u}p$ to mean that in $\Aae$ there is a run on $u$ from an initial
state to $p$.

\begin{defi}
  For a pair $(u,v)$ we define:
  \[
    \itgt(u,v)=\set{q : \exists p.\ \aact{u}p\act{v}q}
  \]
\end{defi}
The general idea is to use $\tgt$ as a measure for how close a pair is to being "pointed", since the "pointed" elements are those for which $\tgt$ is a singleton, as shown in Lemma~\ref{lem:theta-pointed}. Given a pair $(u,v)$ that is not yet "pointed", the algorithm starts by finding an $x$ such that $\theta(u,vxv) \subsetneq \tgt(u,v)$.  The following example illustrates that such an $x$ has to be chosen with care because the length of $vxv$ at least doubles w.r.t.\ the length of $v$.

\begin{figure}
\begin{tikzpicture}[
    state/.style={circle, draw, inner sep=1pt},
    edge label/.style={draw=none, fill=none}
]
\node[state](q00) at (0,0) {$q_0^0$};
\node[state](q01) at (0,1) {$q_0^1$};
\node at (1,0) {$\cdots$};
\node at (1,1) {$\cdots$};
\node[state](qi0) at (2,0) {$q_i^0$};
\node[state](qi1) at (2,1) {$q_i^1$};
\node at (4,0) {$\cdots$};
\node at (4,1) {$\cdots$};
\node[state](qk0) at (5,0) {$q_{k}^0$};
\node[state](qk1) at (5,1) {$q_{k}^1$};

\draw[->](q00) edge[bend left] node[edge label, left]{\scriptsize $a_0$ }(q01);
\draw[->](q01) edge[bend left] node[edge label, right]{\scriptsize $\S_{>0}$ }(q00);

\draw[->](qi0) edge[bend left] node[edge label, left]{\scriptsize $a_i$ }(qi1);
\draw[->](qi1) edge[bend left] node[edge label, right]{\scriptsize $\S_{>i}$ }(qi0);
\draw[->](qi0) edge[loop right] node[edge label,right] {\scriptsize $\S_{<i}$} ();
\draw[->](qi1) edge[loop right] node[edge label,right] {\scriptsize $\S_{<i}$} ();

\draw[->](qk0) edge[bend left] node[edge label, left]{\scriptsize $a_k$ }(qk1);
\draw[->](qk1) edge[bend left] node[edge label, right]{\scriptsize $a_{k+1}$ }(qk0);
\draw[->](qk0) edge[loop right] node[edge label,right] {\scriptsize $\S_{<k}$} ();
\draw[->](qk1) edge[loop right] node[edge label,right] {\scriptsize $\S_{<k}$} ();
\end{tikzpicture}
\caption{The $2$-transitions of the minimal "history-deterministic" co-Büchi automaton for the language from Example~\ref{ex:long-SCC-path} with state set $Q_k$.}\label{fig:long-SCC-path}
\end{figure}

\begin{exa}\label{ex:long-SCC-path} 
  Consider $k \ge 1$, and $\S := \{a_0, \ldots,a_{k+1}\}$. 
For $i \in \{0,\ldots,k\}$ let $\S_{>i} := \{a_j \mid i < j \le
  k+1\}$ and $\S_{<i} := \{a_j \mid 0 \le j < i\}$. 
  The language $L$ contains those $\omega$-words $w$ such that for some $i \in \{0,\ldots,k\}$ only finitely often
  patterns of the form $a_i\S_{<i}^*a_i$ and $\S_{>i}\S_{<i}^*\S_{>i}$ appear in $w$.
  That is, when removing all letters from $\S_{<i}$, there are no two successive
  occurrences of $a_i$ and no two successive occurrences of $\S_{>i}$. For
  example,  $a_j^\omega \in L$ for all $j \in \{0,\ldots,k-1\}$ because there
  are no occurrences of the bad patterns for all $i \in \{j+1, \ldots,k\}$. But
  $a_k^\omega \not\in L$ because for $i \in \{0,\ldots,k-1\}$ the pattern
  $\S_{>i}\S_{>i}$ occurs infinitely often, and for $i = k$, the pattern
  $a_ia_i$ occurs infinitely often. Similarly, $a_{k+1}^\omega \not\in L$.

  Figure~\ref{fig:long-SCC-path} shows the $2$-transitions of the minimal
  "history-deterministic" co-Büchi automaton for $L$. The language has only one
  $\siml$-class and thus the automaton contains all possible $1$-transitions.
  Note that all the components have the same structure: the left-most component
  does not have self-loops because $\S_{<0}$ is empty; and the edge label
  $a_{k+1}$ on the transition from $q_k^1$ to $q_k^0$ corresponds to the only
  letter from $\S_{>k}$.
  There are words making the automaton
  behave as a counter counting up to $2^k$, as explained in the following.
  Letter $a_i$ sets the $i$-th component to $q_i^1$, all components $i'<i$  to
  $q_{i'}^0$, and leaves all other components (with larger index) unchanged. So
  letter $a_i$ is responsible for setting bit $i$ to $1$, all bits for $j<i$ to
  $0$, while leaving all other bits unchanged. The condition defining the language
  says that there is a bit $i$ such that from some point onwards the actions on
  this bit alternate (setting to $1$ and to $0$).

  Assume that we want to construct a "pointed" pair that extends $(\e,v_0)$ with
  $v_0 = \e$. With $x_0 = a_0$ we get $\tgt(\e,v_0x_0v_0) \subsetneq
  \tgt(\e,v_0)$, so we let $v_1 := v_0x_0v_0 = a_0$. Then
  $\tgt(\e,v_1)$ contains all states except $q_0^0$. If we take $x_1 = a_1$, then $\tgt(\e,v_1x_1v_1)$ contains all states except $q_0^0,q_1^0$, so $\tgt(\e,v_1x_1v_1) \subsetneq \tgt(\e,v_1)$ and we let
  $v_2 := v_1x_1v_1 = a_0a_1a_0$.
  If we continue like this, always letting $v_{j+1} = v_jx_jv_j$ with $x_j = a_j$, then $\tgt(\e,v_{j+1})$ contains all states
  except $q_0^0, \ldots, q_j^0$.

  A "pointed" pair constructed in this way would thus be of exponential length because $|v_j| = 2^j -1$. We note that we also get an exponential growth of the words if each $x_j$ is chosen ``greedily'' as the "smallest" word such that $\tgt(\e,v_jx_jv_j) \subsetneq \tgt(\e,v_j)$. We have chosen the $x_j$ as above because then the constructed words have some further properties, as explained below.
  \begin{description}
  \item[C1] $(\e,v_k)$ visits $2^k$ different SCCs of the $\tbe$-graph in the sense that for all prefixes $v \prefeq v' \prefeq v_k$ with $v \not= v'$, we have $(\e,v) \ntbe (\e,v'w)$ for all $w$.
  \item[C2] $(\e,a_kv_k)$ is the shortest in its $\tbe$-class.
  \end{description}
  These two facts show that there are not only exponentially many $\tbe$-classes but also
  classes in which the shortest representative has exponential length in the
  size of $\Aae$.

  In the remainder of the example we prove claims C1 and C2.

  Recall that $\siml$ is trivial, so we can always assume that the first component of pairs is $\e$. For this reason, we can consider $\tbe$ to be a relation on single words instead of pairs by omitting the first component. We first show that for our example language $L$, we have for all $v,v' \in \S^*$:
  \[
  v \tbe v' \text{ iff } \tgt(v) = \tgt(v'). \tag{$*$}
  \]
  If $\tgt(v) = \tgt(v')$, then $\sfl(v) = \bigcup_{q \in \tgt(v)}\Ls(q) = \bigcup_{q \in \tgt(v')}\Ls(q) = \sfl(v')$, and hence $v \tbe v'$. (So this direction holds in general, it does not use any properties of $L$).

  For the other direction, assume $\tgt(v) \not= \tgt(v')$, and wlog.\ let $q_i^h \in \tgt(v) \setminus \tgt(v')$ for $i \in \{0,\ldots, k\}$ and $h \in \{0,1\}$. The only state from which there is a rank 2 run on $a_ia_{k+1}a_ia_{k+1}$ is $q_i^0$, and the only state from which there is a rank 2 run on $a_{k+1}a_ia_{k+1}$ is $q_i^1$. Hence, if $h=0$, then $a_ia_{k+1}a_ia_{k+1} \in \sfl(v) \setminus \sfl(v')$, and if $h=1$, then $a_{k+1}a_ia_{k+1} \in \sfl(v) \setminus \sfl(v')$, and therefore $v \ntbe v'$. This finishes the proof of $(*)$.
  \medskip

  We write $(q_0^0, \ldots, q_i^0) \act{x} (q_0^1, \ldots, q_i^1)$ for expressing that  $q_j^0 \act{x} q_j^1$ for each $j \in \{0,\ldots,i\}$, so $x$ maps the $0$-state to the $1$-state for each component $\le i$.
  \medskip

  \noindent\emph{Proof of claim C2:}
  By induction on $i$, one easily shows that $v_{i+1}$ is the shortest $x$ with $(q_0^0, \ldots, q_i^0) \act{x} (q_0^1, \ldots, q_i^1)$. For $i=0$ we have $v_{i+1} = a_0$, so the claim holds. For $i>0$, such a word $x$ needs to contain $a_i$. Since $a_i$ has no rank 2 transition on $q_j^0$ for $j < i$, and since $q_j^1 \act{a_i} q_j^0$ for all $j < i$, we obtain that $x$ must be of the form $x_1a_ix_2$ with $(q_0^0, \ldots, q_{i-1}^0) \act{x_{1,2}} (q_0^1, \ldots, q_{i-1}^1)$. By induction, the shortest such word is $v_ia_iv_i = v_{i+1}$.

  Now note that $\tgt(a_k) = \{q_0^0, \ldots, q_{k-1}^0,q_k^1\}$, and
  $\tgt(a_kv_k) = \{q_0^1, \ldots, q_{k-1}^1,q_k^1\}$ because $(q_0^0, \ldots,
  q_{k-1}^0) \act{v_k} (q_0^1, \ldots, q_{k-1}^1)$ as explained above.\sloppy As
  every $v$ with $\tgt(v) = \{q_0^1, \ldots, q_{k-1}^1,q_k^1\}$ needs to contain
  $a_k$ followed by an $x$ with $(q_0^0, \ldots, q_{k-1}^0) \act{x} (q_0^1,
  \ldots, q_{k-1}^1)$, and since $v_k$ is the shortest such $x$, we obtain that
  $a_kv_k$ is the shortest word $v$ with $\tgt(v) = \{q_0^1, \ldots, q_{k-1}^1,q_k^1\}$ and thus by $(*)$ also the shortest representative of its class. 
  \smallskip

  \noindent\emph{Proof of claim C1:}
  We show this claim by showing that if a word $x \in \Sigma_{<i}^*$ for $i \in \{0, \ldots,k+1\}$ is such that $\theta(x)$ contains at least one state from each component, then $|x| \le 2^i-1$. Before proving this, let us first see how C1 follows from this. We have $v_k \in  \Sigma_{<k}^*$, and from the proof of the second claim we know that $(q_0^0, \ldots, q_{k-1}^0) \act{v_k} (q_0^1, \ldots, q_{k-1}^1)$, and thus $\tgt(v_k)$ contains a state from each component. Since $\tgt(v_k)$ contains both states from component $k$ (because $v_k \in  \Sigma_{<k}^*$), we obtain that for each prefix $v$ of $v_k$, the class $\ecls{v}$ only contains words over $\Sigma_{<k}$ (using $v \tbe v' \text{ iff } \tgt(v) = \tgt(v')$ by $(*)$), hence only words of length at most $2^k-1$. If there were $v \prefeq v' \prefeq v_k$ with $v \not= v'$ and $v \tbe v'w$ for some $w$, then we obtain a contradiction by showing that $\ecls{v}$ contains words of arbitrary length as follows. We have $v' = vx$ for some $x \not= \e$. From $v \tbe v'w = vxw$, by repeated application of Lemma~\ref{lem:equiv-right}, we get $v \tbe v(xw)^n$ for all $n$. Since $x \not= \e$, this witnesses arbitrary long words in $\ecls{v}$. So C1 follows.

  For completing the proof, let $x \in \Sigma_{<i}^*$ for $i \in \{0, \ldots,k+1\}$ such that $\theta(x)$ contains at least one state from each component. We use induction on $i$. If $i=0$, then $\Sigma_{<i} = \emptyset$, so $x=\e$ and the claim holds. If $i > 1$, then $x$ contains the letter $a_{i-1}$ at most once because otherwise $\tgt(x)$ would not contain a state from component $i-1$. If $x$ does not contain $a_{i-1}$, then by induction $|x| \le 2^{i-1}-1$. Otherwise, $x = x_1a_{i-1}x_2$ with $x_1,x_2 \in \Sigma_{<i-1}^*$, and by induction $|x_1a_{i-1}x_2| \le 2(2^{i-1}-1)+1 = 2^i-1$.
  
  This finishes the proof of the claims and Example~\ref{ex:long-SCC-path}.
\end{exa}

In order to avoid ending up with words of exponential length, we show 
a method of choosing $x$ such that either makes a significant decrease from $\tgt(u,v)$ to
$\tgt(u,vxv)$, or it makes progress without doubling the
second element of the pair.
These two types of progress correspond to, respectively, the first and the
second while loop in Listing~\ref{lst:find-pointed}. 
We need the following definitions that are illustrated in
Example~\ref{ex:long-SCC-path-learn}. 

\begin{defi}\label{def:supported}
  We say:
  \begin{itemize}
    \item \AP$(u,v)$ is \emph{""supported""} if $\tgt(u,v)\not=\es$.
    \item \AP$(u,v)$ is \emph{""double-supported""}  if there are at least two elements from
    $\tgt(u,v)$ in the same "safe SCC" of $\Aae$.
    \item \AP$(u,v)$ is \emph{""single-supported""} if it is supported but not "double-supported".
  \end{itemize}
\end{defi}
Basically, the algorithm $\FindPointedL$ in Listing~\ref{lst:find-pointed} extends the second component of the given pair $(u,v)$ by iteratively constructing words $\bv$ and $w$, starting with
$\bv=v$ and $w=\e$, such that the size of $\tgt(u,w\bv)$ decreases.
The notions of double-supported and single-supported correspond to the two while
loops in Listing~\ref{lst:find-pointed}.
If $(u,\bv)\napproxbot$ is not "pointed" then $\tgt(u,\bv)$ is "supported" but not a
singleton; cf.\ Lemma~\ref{lem:theta-pointed} below.
If $(u,\bv)$ is "double-supported", then by Lemma~\ref{lem:double-supported} there
is $x$ satisfying the condition of the first while loop. 
When the condition of the first while loop is no longer true, then $(u,\bv)$ is "single-supported".
Either $(u,w\bv)$ is "pointed", or by Lemma~\ref{lem:single-supported} there is
$x$ satisfying the condition of the second while loop, and $\tgt(u,w\bv)$ decreases. 
In each iteration of the first while loop, the length of the second component
more than doubles, but Lemma~\ref{lem:x-acceleration} shows that at the same time
$\tgt(u,\bv)$  decreases quickly enough.
In the iterations of the second while loop, the length of the second component increases only by a
polynomial factor. 
Before going into details we explain this on our example.

\begin{exa}
  \label{ex:long-SCC-path-learn}
  We continue Example~\ref{ex:long-SCC-path} from Figure~\ref{fig:long-SCC-path}
  to  illustrate the concepts of "single-supported" and "double-supported" and
  the use of the different lemmas in the construction of "pointed" pairs. Note
  that for this example language $\siml$ is trivial, so $\Rsim = \{\e\}$ and we
  can ignore all conditions on $\siml$ in the algorithm. 

  Consider Listing~\ref{lst:algo}, and, for the purpose of illustration, assume that the component $C_0$ consisting of $\{q_0^0,q_0^1\}$ has already been constructed. Then $x = a_0$ in line~\ref{alg:x}
  because $(a_0)^\omega \in L$ but it is not accepted by the component $C_0$.
  We have $d=2$ in line~\ref{alg:pointed}, and $\FindPointedL(\e,a_0a_0)$ is
  called. We now enter the notation in this call of $\FindPointedL$. For better
  readability we index $\bv$ from different iterations of the loop.

  We start with $\bv_0=a_0a_0$, for which $\tgt(\e,\bv_0) = Q_k \setminus
  \{q_0^0,q_0^1\}$, that are all states except the ones from $C_0$. So
  $(\e,\bv_0)$ is "double-supported", and Lemma~\ref{lem:double-supported}
  guarantees the condition of the while loop in line~\ref{fp:while1} is
  satisfied. Then $x=a_k$ in line~\ref{fp:x1} because $(\bv_0a_k)^\omega \notin
  L$, $\bv_0a_k\bv_0 \napproxbot$. We have $\bv_0(a_k\bv_0)^2 \approxbot$
  because the repetition of $a_0$ removes all states from $C_0$, and the
  repetition of $a_k$ with only $a_0$ in between removes the states of all other
  components. Hence, $m=1$ in line~\ref{fp:m} and we continue with $\bv_1 :=
  \bv_0a_k\bv_0$. Note that $|\bv_1| = 2|\bv_0| + |x|$. Lemma~\ref{lem:x-size}
  guarantees that the length of $x$ is polynomial in the size of $\Aae$, and
  Lemma~\ref{lem:x-acceleration} guarantees that $|\tgt(\e,\bv_1)| \le \frac{1}{2} |\tgt(\e,\bv_0)|$.
And indeed we have $\tgt(\e,\bv_1) = \{q_1^0, \ldots, q_{k-1}^0,q_k^1\}$. So these two lemmas ensure that the length of the $\bv$ constructed in the first loop remains polynomial. For this, it is essential that $u(\bv x)^\omega$ is not in $L$ in line~\ref{fp:while1}, which ensures a quick decrease of $\tgt$, and avoids the problems illustrated in Example~\ref{ex:long-SCC-path}.

At this point the condition of the while loop in line~\ref{fp:while1} is not satisfied anymore
for $(\e,\bv_1)$. To see this, assume that $x$ is such that $\bv_1x\bv_1\napproxbot$, which by Lemma~\ref{lem:theta-empty} is equivalent to $\tgt(\bv_1x\bv_1) \not= \emptyset$. Since it is not possible to read $\bv_1$ from one of the states $\{q_1^0, \ldots, q_{k-1}^0,q_k^1\}$ with only rank $2$ transitions, this means that $\tgt(\bv_1x)$ contains a state from $\{q_1^1, \ldots, q_{k-1}^1,q_k^0\}$. Assume that $q_j^1 \in \tgt(\bv_1x)$ for some $j \in \{1, \ldots, k\}$. Then we obtain $q_j^0 \act{x} q_j^1 \act{\bv_1} q_j^0$, which means that $\bv_1(x\bv_1)^\omega \in L$ because there is an accepting run on it. Similarly, if $q_k^0 \in \tgt(\bv_1)$, we get $\bv_1(x\bv_1)^\omega \in L$ because $q_k^1 \act{x} q_k^0 \act{\bv_1} q_k^1$.
Since $(\e,\bv_1)$ is "single-supported" (but not yet
"pointed" because $\tgt(\e,\bv_1)$ is not a singleton),
Lemma~\ref{lem:single-supported} guarantees that the condition of the while loop
in line~\ref{fp:while2} is satisfied (with $w = \e$) with an $x$ whose length is
polynomial in $\Aae$. Then $x=a_1$ in line~\ref{fp:x2} because
$\tgt(\e,\bv_1a_1\bv_1) = \{q_1^0\}$, so $(\e,\bv_1a_1\bv_1) \napproxbot$ and
$(\e,\bv_1a_1\bv_1) \napproxl (\e,\bv_1)$ because, for example,
$(\e,\bv_1a_1\bv_1a_2) \approxbot$ but $(\e,\bv_1a_2) \napproxbot$.

Note that Lemma~\ref{lem:single-supported} only gives $\tgt(\e,\bv_1a_1\bv_1) \subsetneq \tgt(\e,\bv_1)$, so it could happen that the size of $\tgt$ decreases by only $1$. But every application of the second while loop only increases the length of $w$ by $x\bv$ for the $x$ of the current execution of the while loop, and the fixed $\bv$ resulting from the first loop, which is polynomial in $\Aae$ as argued above (formal arguments are given in the proof of Proposition~\ref{prop:idealized-pointed-finding-correct}). 

  Now $(\e,\bv_1a_1\bv_1) = (\e,a_0a_0a_ka_0a_0a_1a_0a_0a_ka_0a_0)$ is "pointed"
  and returned from the call $\FindPointedL(\e,a_0a_0)$. The main algorithm then
  proceeds to construct the component of the returned "pointed" pair, which
  is $C_1$ in this case. This construction is based on
  Lemma~\ref{lem:component-is-safe}. Moreover, Lemma~\ref{lem:xd-gives-new} indeed
  guarantees that the returned "pointed" pair always belongs to a new component.  
\end{exa}

We proceed with the formal statements and arguments.
The running time of the algorithm and the length of the computed words are expressed in terms of \AP $\iasize:=|\Aae|$ that
is the size of $\Aae$.
We start with two auxiliary lemmas, relating important notions on pairs of words to their $\tgt$-sets.

\begin{lem}\label{lem:theta-empty}
  $\tgt(u,v) = \emptyset$ iff $(u,v) \approxbot$.
\end{lem}
\begin{proof}
Assume that $\tgt(u,v) \not= \emptyset$.  If $v=\e$, then $(u,v) \napproxbot$ by
definition. Otherwise, let $q \in \tgt(u,v)$ and $p$ such that there is a run
$\aact{u} p \act{v} q$, and let $x$ be such that $q \act{x} p$, which exists
since $\Aae$ is normalized (Lemma~\ref{lem:Aae-is-normalized}). Then $uvx \siml u$ and $u(vx)^\omega \in L(\Aae)=L$,
and hence $(u,v) \napproxbot$. 

Now assume that $(u,v) \napproxbot$. If $v=\e$, then $\tgt(u,v) \not= \emptyset$
because $\Aae$ is "unsafe-saturated" (Lemma~\ref{lem:Aae-is-normalized}) and hence there is a run on every word.
Otherwise, let $x$ be such that $uvx \siml u$ and $u(vx)^\omega \in L$. Then
there is a run of the form $\aact{u(vx)^i}p\act{(vx)^j}p$ for some $i \ge0$ and
$j \ge 1$. 
Since $u(vx)^i \siml u$, we get $\aact{u}p$ by "unsafe-saturation" and
$\aact{u}p\act{v}p'$ for some $p'$, and thus $\tgt(u,v) \not= \emptyset$. 
\end{proof}

\begin{lem}\label{lem:theta-pointed}
  $\tgt(u,v)$ is a singleton iff $(u,v)$ is "pointed". 
\end{lem}
\begin{proof}
  For the left-to-right direction, if $\tgt(u,v)$ is a singleton, then $(u,v) \napproxbot$ by Lemma~\ref{lem:theta-empty}. For showing that $(u,v)$ is "pointed", take some $u_1u_2\siml u$ with $(u_1,u_2v)\napproxbot$.
  We need to show that $\sfl(u_1,u_2v)= \sfl(u,v)$. Since $\e$ is in $\sfl$ of every pair that is $\napproxbot$, we only need to consider non-empty words below.

  First we observe that $\sfl(u_1,u_2v)\incl \sfl(u,v)$ without any
  assumption on $\tgt(u,v)$.
  Indeed, if $x\in \sfl(u_1,u_2v)$ then there is $y$ such that
  $u_1(u_2vxy)^\w\in L$ and $u_1u_2vxy\siml u_1$.
  But then $u_1u_2(vxyu_2)^\w\in L$ and $u_1u_2vxyu_2\siml u_1u_2$.
  This shows $x\in\sfl(u,v)$ as $u_1u_2\siml u$. 

  For the other inclusion take $x\in\sfl(u,v)$.
  Then there is $y$ such that $u(vxy)^\w\in L$ and $uvxy\siml u$.
  This implies that we have a run $\aact{u(vxy)^i}p\act{v}q\act{xy}$ for some $p,q$ and $i \ge 0$, and thus also $\aact{u}p\act{v}q\act{xy}$ by "semantic-determinism" and "unsafe-saturation".
  Since $(u_1,u_2v)\napproxbot$, again using "semantic-determinism" and "unsafe-saturation", we have
  $\aact{u_1}p_1\act{u_2}p_2\act{v}q_1$ for some $p_1,p_2$ and $q_1$.
  But since $\tgt(u,v)$ is a singleton, and $u_1u_2\siml u$, we conclude that
  $q_1=q$.
  Thus, $x\in\sfl(u_1,u_2v)$ because
  $\aact{u_1}p_1\act{u_2}p_2\act{v}q\act{x}$, and $\Aae$ is "normalized" (Lemma~\ref{lem:Aae-is-normalized})

  For the right-to-left direction we suppose $\tgt(u,v)$ is not a singleton, and
  show that $(u,v)$ is not "pointed".
  As $\tgt(u,v)$ is not a singleton, in $\Aae$ we have runs:
  \[
    \aact{u}p_1\act{v}q_1\quad\text{and}\quad \aact{u}p_2\act{v}q_2
  \]
  with $q_1\not=q_2$.
  Since $L(q_1)=L(q_2)$ by "semantic-determinism", we must have $\Ls(q_1)\not=\Ls(q_2)$ by Lemma~\ref{lem:states-in-Aae}.
  Suppose $z\in \Ls(q_1)-\Ls(q_2)$.
  Since $\Aae$ is "normalized", there is $z'$ such that $q_1\act{zz'}q_1$.
  Take $x=zz'$.
  We have $(u,vx)\napproxbot$ by Lemma~\ref{lem:theta-empty} because $p_1\act{vx}q_1 \in \tgt(u,vx)$. Let $z_{p_2}$ be a "central sequence" for $p_2$, which exists by Lemma~\ref{lem:Aae-central}. Then $uz_{p_2}\siml u$, and $(u,z_{p_2}v)\napproxbot$ because $p_2\act{z_{p_2}v}q_2$.
  On the other hand, $(u,z_{p_2}vx)\approxbot$ because for every state $p'$ with $L(p)=u^{-1}L$ we have
  that either $p'\bact{z_{p_2}}$ or $p'\act{z_{p_2}}p_2\act{v}q_2\bact{z}$.
  So $\sfl(u,v)\not=\sfl(u,z_{p_2}v)$, hence $(u,v)$ is not "pointed".
\end{proof}

The next two Lemmas~\ref{lem:double-supported} and \ref{lem:x-acceleration} deal
with the condition of the first while loop in line~\ref{fp:while1} of
$\FindPointedL$, as discussed in  Example~\ref{ex:long-SCC-path-learn}.

\begin{lem}\label{lem:double-supported}
  If $(u,v)$ is "double-supported" then there is $x\in\S^+$ such that $u\siml uvx$,
  $(u,vxv) \napproxbot$, and $u(vx)^\w\not\in L$.
\end{lem}
\begin{proof}
  If $\tgt(u,v)$ is "double-supported" then there are $p,p'\in
  \tgt(u,v)$ and runs $\aact{u}p_u\act{v}p$, $\aact{u}p'_u\act{v}p'$, where $p_u,p'_u,p,p'$ are in
  the same safe SCC of $\Aae$.
  By Lemma~\ref{lem:states-in-Aae} we can assume that there is $y\in \S^+$ with $p\act{y}q$, for some
  $q$, and $p'\bact{y}$.
  Then we can take some $w$ with $q\act{w}p'_u$ as $q$ and $p'_u$ are in the same
  safe SCC of $\Aae$.
  Consider:
  \[
    \aact{u}p_u\act{v}p\act{y}q\act{z_q}q\act{w}p'_u\act{v}p'\bact{y}
  \]
  where $z_q$ is a "central sequence" for $q$, which exists by Lemma~\ref{lem:Aae-central}.
  Take $x=yz_qw$.
  We have $u\siml uvx$ because there is a run on $uvx$ ending in $p'_u$ and
  $\Aae$ is "semantically-deterministic" (Lemma~\ref{lem:Aae-is-normalized}).
  We also have $u(vx)^\w\not\in L$ because for every $p_u''$ with $\aact{u}p_u''$ we
  have either $p_u''\bact{vx}$ or $p_u''\act{vx}p'_u\act{v}p'\bact{x}{}$.
\end{proof}

\begin{lem}\label{lem:x-acceleration}
  Suppose $(u,v)\napproxbot$, $u\siml uv\siml uvx$, and
  $u(vx)^\w\not\in L$.
  Let $m$ maximal such that $(u,v(xv)^m)\napproxbot$.
  Then $\tgt(u,v(xv)^m)\subseteq \tgt(u,v)$ and $|\tgt(u,v(xv)^m)|\leq \frac{1}{m+1}|\tgt(u,v)|$.
\end{lem}
\begin{proof}
  First note that $m$ is well-defined because we have $(u,v(xv)^0) = (u,v)
  \napproxl \bot$, and further if $(u,v(xv)^n)\napproxl \bot$ for all $n$, then
  $\tgt(u,v(xv)^n) \not= \emptyset$ for all $n$ by Lemma~\ref{lem:theta-empty}, and hence $\Aae$ would accept $u(vx)^\w$ as $\Aae$ is "safe-deterministic".
  Consider:
  \[
    \aact{u}P_u\act{v}Q_0\act{xv}Q_1\act{xv}\cdots\act{xv}Q_m\act{xv}Q_{m+1}=\es\ .
  \]
  Here $P_u = \tgt(u,\e)$ is the set of states $\ecls{u',v'}$ with $(u',v')\siml u$, and $Q_i = \tgt(u,v(xv)^i)$ for $i \in \{0,\ldots,m+1\}$, where $Q_{m+1} = \emptyset$ follows from the choice of $m$ and Lemma~\ref{lem:theta-empty}.
  Thanks to $u\siml uv\siml uvx$, we have  $Q_i\incl P_{u}$, for all
  $i$.
  Further, we get $Q_{i+1}\incl Q_i$ as follows. For $i=0$ we have $P_u \act{v} Q_0\act{x}Q'_0\incl P_u$ because $u\siml uvx$, and $Q_0'\act{v}Q_1$. This implies $Q_1 \subseteq Q_0$. Then we can proceed by induction because $Q_{i-1} \act{xv} Q_i \act{xv}Q_{i+1}$ and $Q_{i-1} \subseteq Q_i$ imply $Q_{i+1} \subseteq Q_i$.
  This shows that $\tgt(u,v(xv)^m) = Q_m \subseteq Q_0 = \tgt(u,v)$.

  For the claim on the size of $\tgt(u,v(xv)^m)$, let $R_i=Q_i-Q_{i+1}$, for $i=0,\dots,m$, and $R_{m+1} = \emptyset$.
  Since $Q_0 \supseteq Q_1 \supseteq \cdots \supseteq Q_m \supseteq Q_{m+1} = \emptyset$, we have $Q_0=R_0\cup R_{i+1}\cup\dots\cup R_m$ and the union is disjoint.
  We show that $R_i \act{xv} R_{i+1}'$ with  $R_{i+1}' \supseteq R_{i+1}$ for $i =0,\dots,m$. For $i = m$ note that $Q_m=R_m$ and $R_m\act{xv}\es = R_{m+1}$. For $i \in \{0, \ldots, m-1\}$,
  we have $Q_i = R_i \cup Q_{i+1}$ and $Q_{i+1} \act{xv} Q_{i+2}$, and hence $Q_i \act{xv} (R_{i+1}' \cup Q_{i+2}) = Q_{i+1}$. Since $R_{i+1} = Q_{i+1} - Q_{i+2}$, we get $R_{i+1}' \supseteq R_{i+1}$. 
  
  By "safe-determinism" we thus have $|R_i|\geq |R_{i+1}|$. So $|R_m|\leq\frac{1}{m+1}|Q_0|$
  as $Q_0=R_0\cup\dots\cup R_m$ and all the $R_i$ are pairwise disjoint. 
  So the claim of the lemma follows from $R_m=\tgt(u,v(xv)^m)$ and $Q_0 = \tgt(u,v)$.
\end{proof}

The next two Lemmas~\ref{lem:Aae-central-size} and \ref{lem:x-size} deal with
the size of a witness for the first while loop in line~\ref{fp:while1} of
$\FindPointedL$, whose existence is guaranteed by
Lemma~\ref{lem:double-supported} for double supported pairs.

\begin{lem}\label{lem:Aae-central-size}
  Every state of $\Aae$ has a "central sequence" of size $<\asize^3$.
\end{lem}
\begin{proof}
  As first step, consider a state $q$ that is maximal in a sense that there is no $p$ such
that $L(p)=L(q)$ and $\Ls(q)\incl \Ls(p)$.
Let $p_1,\dots,p_k$ be an arbitrary enumeration of states such that $L(p_i)=L(q)$.
By the choice of $q$, for every $p_i$ there is $y_i$ such that $q\act{y_i}q'_i$ for some state $q_i'$,
and $p_i\bact{y_i}$. Since $\Aae$ is "normalized" (Lemma~\ref{lem:Aae-is-normalized}), $q$ and $q'_i$ are in the same safe SCC of $\Aa$ so
there is $q'_i\act{y'_i} q$. 
Taking $z_i=y_iy'_i$ we have $q\act{z_i}q$ and $p_i\bact{z_i}$.
The length of $y_i$ can be bounded by $\asize^2$, by considering the product of the
automaton with itself after removing all $1$-transitions.
The length of $y'_i$ can be bounded by $\asize-1$.
So the length of $z_i$ is bounded by $\asize^2+\asize-1$.

Now let us look at our fixed enumeration $p_1,\dots,p_k$.
If $z_1$ is "central" for $q$ then we are done.
Otherwise,  consider the
smallest $i_1$ such that $p_{i_1}\act{z_1} p'$ for some state $p'$. 
Observe that by "semantic-determinism" (Lemma~\ref{lem:Aae-sd-and-us}), $L(q)=L(p_{i_1})=L(p')$.
Hence, $p'=p_{j_1}$ for some $j_1$.
Clearly $i_1>1$.
If $z_1z_{j_1}$ is not "central", we can find the smallest $i_2$ such that
$p_{i_2}\act{z_1z_{j_1}} p_{j_2}$ for some $j_2$.  
Clearly $i_2>i_1$ as $p_{i_1}\bact{z_1z_{j_1}}$.
Continuing like this we get the required $z_q$. 
Since we do at most $k$ steps in this construction and $k$ is bounded by
$\asize-1$ we have that the size of $z_q$ is bounded by $(\asize-1)(\asize^2+\asize-1)=\asize^3-2\asize+1$. 

In the second step of the construction consider some $q'$ that is not maximal.
Hence, there is some $q$ with $L(q')=L(q)$ and $\Ls(q')\incl \Ls(q)$. Choose $q$ such that $\Ls(q)$ is maximal for inclusion. By Lemma~\ref{lem:states-in-Aae}, there is no other $p$ such
that $L(p)=L(q)$ and $\Ls(q)\incl \Ls(p)$, hence $q$ is maximal in the sense of
the first step, so there is a "central sequence" $z_q$ for $q$.
By Lemma~\ref{lem:Aae-central} we know that $q'$ has a "central sequence"
$z_{q'}$. 
Since $\Ls(q')\incl \Ls(q)$ we have $q\act{z_{q'}}q'$.
Since $\Aae$ is "normalized" (Lemma~\ref{lem:Aae-is-normalized}), $q$ and $q'$ are in the same safe SCC of $\Aa$.
So we have $q'\act{x} q\act{y} q'$ for some $x$ and $y$ with size of $xy$ at most $2\asize-2$.
It is easy to check that $xz_qy$ is also "central" for $q'$.
Given that the size of $z_q$ is at most $\asize^3-2\asize+1$ we get the desired
bound $\asize^3$.
\end{proof}

\begin{lem}\label{lem:x-size}
    If there is an $x$ such that $u\siml uvx$, $(u,vxv)\napproxbot$ and
  $u(vx)^\w\not\in L$ then there is one of size $<2\asize^3$.
\end{lem}
\begin{proof}
  We claim that the assumptions of the lemma imply that in $\Aae$ there is the
  following run for some $i \ge 0$:
  \[
    \aact{u}p_u\act{(vx)^i}p\act{v}q\act{x}p'\act{v}q'\bact{xv}
  \]
  The existence of such a run up to $q'$ follows from $(u,vxv)\napproxbot$.
  And $i$ can be chosen such that $q'\bact{xv}$ because $u(vx)^\w\not\in L$.

  Since $q\act{xv}q'$ and $q'\bact{xv}$, we can choose a shortest $y$ with this property $q\act{y}q'$ and $q'\bact{y}$, which is of length at most $\asize^2$. Let further $y'$ be shortest with $q' \act{y'} p'$, which exists since $\Aae$ is "normalized" (Lemma~\ref{lem:Aae-is-normalized}) and is of size $<\asize$. 
Then $x':=yz_{q'}y'$ for a "central sequence" $z_{q'}$ of $q'$ satisfies the claim of the lemma, as argued in the following. Indeed, $u\siml uvx'$ because $\aact{u}p\act{v}q\act{x'}p'$. The fact $\aact{u}p$ follows from $\aact{u}p_u\act{(vx)^i}p$ and $u(vx)^i \siml u$. Further, $\aact{u}p\act{vx'v}q'$, so $(u,vx'v) \napproxbot$. Finally, $u(vx')^\w\not\in L$ because for every $p''$ with $\aact{u}p''$ we
  have either $p''\act{vy}q''\bact{z_{q'}}$, or
  $p''\act{vy}q''\aact{z_{q'}}q'\act{y'}p'\act{v}q'\bact{y}$.

 The existence of $z_{q'}$ and an $\asize^3$ bound on its size come from
  Lemma~\ref{lem:Aae-central-size}.
  Together with the bounds on $y$ and $y'$ we have $|x'|<2\asize^3$.
\end{proof}

Lemma~\ref{lem:single-supported} below deals with the condition of the second
while loop in line~\ref{fp:while2} of $\FindPointedL$; see an explanation after
Definition~\ref{def:supported}, and Example~\ref{ex:long-SCC-path-learn} for an
illustration. 

\begin{lem}\label{lem:single-supported}
    Suppose $(u,v)$ is "single-supported", $uw\siml u$, and $(u,wv)\napproxbot$
  is not "pointed".
  There is $x\in\S^+$ such that $u\siml uwvx$ and $\bot\napproxl
  (u,wvxv)\napproxl (u,wv)$. Moreover, for every such $x$, $\tgt(u,wvxv)\subsetneq
  \tgt(u,wv)$.
  Additionally, there is such $x$ of size $<2\asize^2$.
\end{lem}
\begin{proof}
  We first show that $\tgt(u,wvxv)\incl\tgt(u,wv)$ for every $x$ such that $u\siml uwvx$.
  For $q\in\tgt(u,wvxv)$ we get $\aact{u}p\act{w}q_1\act{v}q_2\act{x}q_3\act{v}q$.
  By $uw\siml u$ we get $q_2\in \tgt(u,v)$, and by $u\siml uwvx$ we get
  $q\in\tgt(u,v)$.
  But since $(u,v)$ is "single-supported", $q=q_2$.
  Since $q_2\in\tgt(u,wv)$ we get $q\in\tgt(u,wv)$.
  
We now show by contraposition that the additional condition $(u,wvxv)\napproxl (u,wv)$ implies strict inclusion $\tgt(u,wvxv)\subsetneq \tgt(u,wv)$. So assume that $\tgt(u,wvxv) = \tgt(u,wv)$ and let $y \in \Sigma^*$ be such that $uwvy \siml u$. We need to show that $u(wvy)^\omega \in L \Leftrightarrow u(wvxvy)^\omega \in L$ for concluding that $(u,wvxv)\approxl (u,wv)$. If any of the two words is in $L$, then there must be an accepting run starting with $\aact{u}p\act{w}q_1\act{v}q_2$. Then $q_2 \in \tgt(u,wv)$ and thus $q_2 \in  \tgt(u,wvxv)$. As already argued above, we then have $q_2 \act{xv} q = q_2$ because $(u,v)$ is "single-supported". So there is an $xv$-loop with rank $2$ transitions on $q_2$. Now the accepting run on either $u(wvy)^\omega$ or $u(wvxvy)^\omega$ continues from $q_2$ with $ywv$. Since $uwvy \siml u$ and $(u,v)$ is "single-supported", we get that $q_2 \act{ywv} q_2$ (if it would reach a different state, we would get a second state in $\tgt(u,v)$ that is in the same safe SCC as $q_2$). But this means that there is also a $ywv$-loop on $q_2$ with rank $2$ transitions, and hence both $u(wvy)^\omega$ and $u(wvxvy)^\omega$ are accepted if one of them is.  

  We now show the existence of an $x$ as claimed in the statement of the lemma.
  Since $(u,wv)$ is not "pointed", $\tgt(u,wv)$ is not a singleton, by
  Lemma~\ref{lem:theta-pointed}.
  So there are two distinct $q,q'\in\tgt(u,wv)$. 
  They must be in different "safe SCCs" as $(u,v)$ is "single-supported". This implies that $\Ls(q) \not\subseteq \Ls(q')$ by Lemma~\ref{lem:Aae-central} (a "central sequence" for $q$ witnesses this non-inclusion). Hence, there is $y \in \Sigma^+$ with $q\act{y}$ and  $q'\bact{y}$.
  By "normality" (Lemma~\ref{lem:Aae-is-normalized}) there is $y'$ such that $q\act{yy'}p$, where $p$ is from a run witnessing $q \in \tgt(u,wv)$ as follows:
  \[
    \aact{u}p_u\act{w}p\act{v}q\act{yy'}p
  \]
  We take $x=yy'$. Note that $y$ can be chosen of length at most $\asize^2$ and $y'$ of length at most $\asize-1$. So the size of $x$ is bounded by $\asize^2+\asize-1$ that is at most $2\asize^2$. We now show that $x$ satisfies the desired properties. 
  First, $q\in\tgt(u,wvxv)$ so $(u,wvxv)\napproxbot$ by Lemma~\ref{lem:theta-empty}.
  Then $u\siml uwvx$ because the run ends in $p$ and $uw \siml u$.
  It remains to show that $(u,wvxv) \napproxl (u,wv)$. We argue that $q' \in \tgt(u,wv) \setminus \tgt(u,wvxv)$. Assume to the contrary that $q'\in\tgt(u,wvxv)$.
  Then there is a run $\aact{u}p''_u\act{w}p''\act{v}q''\act{yy'v}q'$, for some
  $p''_u,p''$ and $q''$. Then $q''$ is in the same "safe SCC" as $q'$, and thus, $q'' = q'$ because $q', q'' \in \tgt(u,wv)$ and $(u,wv)$ is "single-supported". Since $q'' \act{y}$ and $q' \bact{y}$, we get a contradiction. Hence, $q' \not\in \tgt(u,wvxv)$ as desired. Now take a "central sequence" $z_{q'}$ for $q'$. Then $q' \in \tgt(u,wvz_{q'}) \not= \emptyset$ but $\tgt(u,wvxvz_{q'}) = \emptyset$, so $(u,wvz_{q'}) \napproxbot$ and $(u,wvxvz_{q'}) \approxbot$ by Lemma~\ref{lem:theta-empty}. Hence $(u,wvxv) \napproxl (u,wv)$.
\end{proof}

We are ready to prove that the algorithm $\FindPointedL$ in
Listing~\ref{lst:find-pointed} is correct, runs in polynomial time, and produces
a pointed pair of polynomial size.
\begin{prop}\label{prop:idealized-pointed-finding-correct}
  Given $(u,v)\napproxbot$ with $v \not= \e$ such that $u\siml uv$, $\FindPointedL(u,v)$ in
  Listing~\ref{lst:find-pointed} returns a "pointed" element $(u,w\bv)\napproxbot$. Moreover, the algorithm runs
  in time polynomial in $|\Aae|$, $|w\bv|\leq 2|\Aae|^2|v|+4|\Aae|^6$ and the initial $v$ is a
  prefix of $w\bv$.
\end{prop}  
\begin{proof}
  Recall that we use $\asize$ for the size of $\Aae$.
  If $(u,v)$ is not "pointed" then 
  $\tgt(u,v)$ is not a singleton, by Lemma~\ref{lem:theta-pointed}.
  
  The first case is when $(u,v)$ is "double-supported".
  At the beginning we consider $\bv=v$.
  By Lemma~\ref{lem:double-supported} there
  is $x$ satisfying the condition of the first while loop. 
  In the loop $\bv$ is changed to $\bv(x\bv)^m$ for some $m$. 
  Observe, that $u\siml u\bv(x\bv)^m$, and $(u,\bv(x\bv)^m)\napproxbot$ by the
  choice of $m$, so we can repeat the reasoning for $(u,\bv)$
  where $\bv$ is changed to $\bv(x\bv)^m$.
  Let us note that $m\geq 1$, so indeed $\bv$ is prolonged in this step and
  $\tgt(u,\bv)$ decreases.
  
  When the condition of the first while loop is no longer true then $(u,\bv)$
  is "single-supported" (Lemma~\ref{lem:double-supported}).
  Either $(u,w\bv)$ is "pointed", or by Lemma~\ref{lem:single-supported} there
  is $x$ satisfying the condition of the second while loop.
  In this case we change $w$ to $w\bv x$, and Lemma~\ref{lem:single-supported}
  says that $\tgt(u,w\bv x\bv)$ is strictly included in $\tgt(u,w\bv)$.
  So $(u,w\bv x\bv)$ is "single-supported", and we can repeat the reasoning
  after updating $w$ to $w\bv x$. Thus, the algorithm terminates with a pointed pair $(u,w\bv)$. 
  By examining how $\bv$ and $w$ are constructed we can see that the initial $v$
  is always a prefix of $w\bv$. 
  This shows the last property stated in the proposition.

  It remains to check that the algorithm runs in polynomial time and that the
  size of $w\bv$ is bounded as stated in the proposition.
  Let $v_i$ denote the value of $\bv$ just after the $i$-th iteration of the
  first while loop in line~\ref{fp:while1}.
  Let $m_i$ be the value of $m$ at the same moment.
  We use $v_0$ for the initial $\bv$.
  By Lemma~\ref{lem:x-acceleration} we have 
  \begin{equation}\label{eq:tgt}
    |\tgt(u,v_{i})|\leq \frac{1}{m_i+1}|\tgt(u,v_{i-1})|\ .
  \end{equation}
  As we have observed above, $m_i\geq 1$, so we have in particular 
  $|\tgt(u,v_{i+1})|< |\tgt(u,v_i)|$, for all $i$.
  Thus, the number $k$ of iterations of the first while loop is at most $\asize$. 
  As $x$ used in this loop comes from Lemma~\ref{lem:x-size}, its size is at
  most $\b:=2\asize^3$ giving us $|v_{i}|\leq (m_i+1)|v_{i-1}| + m_i\b$.
  We check that for all $i=1,\dots,k$:
  \begin{equation}\label{eq:bound}
    |v_i|\leq \left(\textstyle\prod_{j=1}^{i}(m_j+1)\right)(|v_0|+i\b)
  \end{equation}
  This holds for $i=1$ as $|v_1|\leq (m_1+1)|v_0| + m_1\b$.
  By induction for $i>1$ we have
  \begin{align*}
    |v_{i}|&\leq (m_{i}+1)|v_{i-1}| + m_{i}\b\\
    &\leq (m_{i}+1)\left(\textstyle\prod_{j=1}^{i-1}(m_j+1)\right)(|v_0|+(i-1)\b) + m_{i}\b\\
    &= \left(\textstyle\prod_{j=1}^{i}(m_j+1)\right)(|v_0| +(i-1)\b) + m_i\b\\
    &\leq \left(\textstyle\prod_{j=1}^{i}(m_j+1)\right)(|v_0| +i\b)
  \end{align*}
  The last inequality holds as $m_i\leq \prod_{j=1}^{i}(m_j+1)$.

  By iterated application of equation~\eqref{eq:tgt} we have $\tgt(u,v_k)\leq
  \frac{\tgt(u,v_0)}{\prod_{i=1}^{k}(m_i+1)}$.
  Since $|\tgt(u,v_k)| \ge 1$, we get $\prod_{i=1}^{k}(m_i+1)\leq \tgt(u,v_0)\leq \asize$.
  Plugging this bound into~\eqref{eq:bound} for $v_k$, we get $|v_k|\leq
  \asize(|v_0|+\asize\b)$. 
  Finally, since $\b=2\asize^3$ and $v_k$ is the result of the first while loop, we get $|\bv|\leq \asize|v_0|+2\asize^5$ at the end of the first while loop.
  
  Now let us turn to the second while loop.
  By Lemma~\ref{lem:single-supported},  $|\tgt(u,w\bv x\bv)|<|\tgt(u,w\bv)|$,
  hence the loop executes at most $\asize$ times.
  Each time $x$ is found thanks to Lemma~\ref{lem:single-supported}, so $|x|\leq 2\asize^2$. 
  Hence, at each iteration of the loop we prolong $w$ by at most $|\bv|+2\asize^2$.
  This gives us a bound on the final size of $w$ as $|w|\leq
  \asize(|\bv|+2\asize^2)$.
  Plugging in our bound on $|\bv|$ from the previous paragraph we obtain
  $|w|\leq \asize(\asize|v_0|+2\asize^5+2\asize^2)$.
  So finally, $|w\bv|\leq \asize^2|v_0|+2\asize^6+2\asize^3+\asize|v_0|+2\asize^5$ that we can bound by
  $2\asize^2|v_0|+4\asize^6$. 
\end{proof}


\subsection{Passive learning algorithm} \label{sub:passive-learner}

The true learning algorithm does not have access to the language $L$, but only
to a "sample" $S$. 
In this section we show how to implement the idealized algorithm from
Listing~\ref{lst:algo} in a passive learning algorithm $\Learn(S)$ presented in
Listing~\ref{lst:palgo}.
It is parameterized by a "sample" $S$ given on the input.
Each line of the idealized algorithm is implemented by a corresponding
function depending on $S$. 
Compared to the idealized algorithm from Listing~\ref{lst:algo}, the algorithm
from Listing~\ref{lst:palgo} has abort statements in
lines~\ref{palg:if-bot},~\ref{palg:x} and~\ref{plag:if-lang}.  
These are needed in case a "sample" has some incomplete information about $L$ that
makes our algorithm return incoherent results. In this case the algorithm returns
some default automaton that accepts precisely $S^+$. Such an automaton is easy
to construct, so we do not detail it here.
This choice of a default automaton, together with the termination condition of the algorithm, guarantees that the algorithm is a "consistent" learner.

The result we prove is:
\begin{thm} \label{thm:learner}
  $\Learn(S)$ is a "consistent" polynomial time "passive learner" that learns the
  minimal history-deterministic co-Büchi automaton $\Aae$ for each co-Büchi
  language $L$ in the limit from "polynomial data". 
\end{thm}

Compared to the idealized algorithm, the passive learning algorithm has to
deduce information about the language $L$ from a "sample". So all the tests that
refer to the language $L$, like membership or equality, have to be relativized
to the "sample" $S$. This is done by the following definitions. 
\begin{itemize}
  \item \AP$(u,v)\iins L$ if $(u,v)\in S^+$. We sometimes write $uv^\w\iins L$
  for readability 
  \item \AP$(u,v)\inins L$ if $(u,v)\in S^-$; also written as $uv^\w\inins L$.
  \item \AP$L\ineqs L(\Aa)$ for some automaton $\Aa$ if there is $(u,v)$
  with either  $(u,v)\ins L$ and $uv^\w\not\in L(\Aa)$, or $(u,v)\nins L$ and $uv^\w\in L(\Aa)$.
  \item \AP$x\insims y$ if there is $(u,v)$ such that either $(xu,v)\ins L$ and 
  $(yu,v)\nins L$, or  $(xu,v)\nins L$ and 
  $(yu,v)\ins L$.
\end{itemize}
The most important thing to notice, is that $(u,v)\nins L$ is not the negation
of $(u,v)\ins L$.
Indeed, it may happen that neither $(u,v)\ins L$ nor $(u,v)\nins L$ hold as
there is no information about $(u,v)$ in the "sample". 
For this reason the formulation of these definitions is so long, as all the
tests should be used positively.
This guarantees that these definitions are stable under extensions:
if they hold for $S$ then they hold for every $S'$ "extending" $S$. 

\begin{lstlisting}[caption={Passive learning algorithm for $L$},label=lst:palgo,float,frame=lines,escapechar=|]
function $\iLearn(S)$:
  $R:=\FindR(S)$
  $NT:=\FindNT(R,S)$ |\label{palg:lineE}|
  forevery $u\in R-NT$ add $C(u,\e)$ to $\Cc$.
  $\Aa=\automaton(R,\Cc,S)$
  while $L\neqs L(\Aa)$ do |\label{palg:while}|
    $u:=\Findu(\Aa,NT,R,S)$; abort if $u=\bot$|\label{palg:u}\label{palg:if-bot}|
    $x:=\Findx(u,\Aa,R,S)$; abort if $x=\bot$ |\label{palg:x}|
    $d:=\max(\set{|C| :  C\in \Cc}\cup\set{1})$ |\label{palg:d}|
    $(u,v):=\FindPointedS(u,x^d,R,S)$; |\label{palg:pointed}|
    $C:=\construct(u,v,R,S)$ |\label{palg:C}|
    add $C$ to $\Cc$
    $\Aa'=\automaton(R,\Cc,S)$; if $L(\Aa)\subsetneq L(\Aa') \incl S^+$ then $\Aa:=\Aa'$ else abort |\label{plag:if-lang}|
  return($\Aa$)
\end{lstlisting}

The passive learning algorithm from Listing~\ref{lst:palgo} follows line by line
the idealized algorithm from Listing~\ref{lst:algo}. 
Each line of the idealized algorithm is replaced by a function call. 
For each function in Listing~\ref{lst:palgo} we show that there is $S$ polynomial in the
size of $\Aae$, for which the function computes the right answer, meaning the
answer computed in the corresponding line of the idealized algorithm.
Moreover, there is such $S$ that is "stabilizing" in the following sense.

\begin{defi}
	A "sample" $S$ that is "consistent" with $L$ is ""$L$-stabilizing"" for a function $Func$ if for every $S'$ "consistent" with $L$ such that $S'\supseteq S$, we
	have $Func(S')=Func(S)$. We just say stabilizing in the following since $L$ is clear from the context.
\end{defi}

\begin{lstlisting}[caption={Finding $\Rsim$},label=lst:R,float,frame=lines]
function $\iFindR(S)$: 
  $R:=\set{\es}$
  repeat
    $D:=\set{x : \forall y\in R.\ x\nsims y}$
    if $D\not=\es$ then add $\min(D)$ to $R$
  until $D=\es$
  return($R$)
\end{lstlisting}

The notion of "stabilizing" $S$ can be seen in action already in the first line of
the algorithm from Listing~\ref{lst:palgo}.
The goal of function $\FindR(S)$ is to compute $\Rsim$, the set of minimal
representatives of all $\siml$ classes.
The algorithm successively adds to $D$ the  "smallest" word that
is in the relation $\nsims$  to all the words already in $D$, as already explained on page~\pageref{def:Rsim}.

\begin{restatable}{lem}{lemFindR}\label{lem:FindR}
	There is a "stabilizing" $S$ of size polynomial in $|\Aae|$ such that
	$\Rsim=\FindR(S)$.  Moreover, the algorithm runs in time polynomial in
	$S$.
\end{restatable}
\begin{proof}
	Consider $\Rsim=\set{x_1,\dots,x_k}$ where the $x_i$ are the minimal representatives of the $\siml$-classes. The length of each $x_i$ and the number $k$ is bounded by the number of states of $\Aae$.  
	For every $x,y\in \Rsim$ with $x \not= y$ there are $u,v$ such that
	$xuv^\w\in L \iff yuv^\w\not\in L$.
	These are at most quadratically many witnesses that are required for pairwise
	separation of all elements from $\Rsim$. We show that their length can be
	bounded polynomially.        
        
        Wlog.\ we can assume that $xuv^\omega \in L$ and $yuv^\omega \notin L$. Further, we can assume that $xu \siml xuv$ and $yu \siml yuv$ and that there is some state $p$ such that $\aact{xu} p \act{v} p$. If this is not the case, then we can replace $u$ by $uv^m$ and $v$ by $v^n$ for appropriately chosen $m,n$ such that the conditions are satisfied. The length of $u$ can then be bounded by the product of $\siml$ with itself ($u$ just needs to reach a certain pair of $\siml$-classes from $x,y$). 
        Further, consider the set $Q_{yu}$ of states that is reachable in $\Aae$ via $yu$ from an initial state. Because $yuv^\omega \notin L$, for each state $q \in Q_{yu}$ there is $m$ such that $q \bact{v^m}$. In particular, the set
        \[
        W_{p,q} = \{w \in \Sigma^* \mid xu \siml xuw, yu \siml yuw, p \act{w} p, \text{ and } q \bact{w}\}
        \]
        is not empty, and it can be accepted by a DFA of size quadratic in $\Aae$.
        Now consider a word $v' = v_1\cdots v_n$ constructed as follows. Start with an arbitrary state $q_1 \in Q_{yu}$ and let $v_1 \in W_{p,q_1}$. Let $Q_{yu}^1$ be the set of states with $Q_{yu} \act{v_1} Q_{yu}^1$. Since $yu \siml yuv_1$ and $q_1 \bact{v_1}$, we get that $Q_{yu}^1$ is a strict subset of $Q_{yu}$ by "safe-determinism" and "semantic-determinism" of $\Aae$. If $Q_{yu}^1 \not= \emptyset$, let $q_2 \in Q_{yu}^1$ and choose $v_2 \in W_{p,q_2}$. We continue the construction until $Q_{yu}^n = \emptyset$, and let $v' :=v_1 \cdots v_n$.
        Then $v' \in W_{p,q}$ for all $q \in Q_{yu}$, and thus $xu(v')^\w \in L$ and $yu(v')^\omega \notin L$.  Note that each $v_i$ can be chosen of length quadratically in $\Aae$, so the length of $v'$ is at most cubic in $\Aae$ and thus polynomial.

        Let the  "sample" $S$ contain separating witnesses of polynomial length for all $x,y \in \Rsim$ with $x \not= y$.
	On the "sample" $S$, the algorithm returns $\Rsim$. 
	On every bigger "sample" it must also return $\Rsim$ due to the minimality of elements
	in $\Rsim$.
  The algorithm runs in polynomial time because the set $D$ is bounded polynomially in the size of $S$ since $D$ consists of prefixes from words in $S$ (by definition of $\nsims$).
\end{proof}

Assuming we have $\Rsim$, we can define further tests on $S$:
\begin{itemize}
  \item  \AP$x\isims y$ if there is $u\in \Rsim$ such that for all $u'\in
  \Rsim$, $u'\not=u$, we have $x\nsims u'$ and $y\nsims u'$ \AP
  \item $(u,v)\inapproxsbot$ if $v = \e$ or there is $x\in\S^*$ such that $uvx \sims u$ and $(u,vx)\ins L$. 
  \item \AP$(u,v)\inapproxs (u',v')$ if either $u\nsims u'$, or $uv\nsims
  u'v'$, or there is $x\in\S^*$ such that $uvx\sims u$ and either $(u,vx)\ins
  L$ and $(u',v'x)\nins L$, or $(u,vx)\nins L$ and $(u',v'x)\ins  L$.
\end{itemize} 
These definitions are also stable under extensions, namely if $S$ is big enough
for $\Rsim$ to have representatives of all $\siml$ classes, and has the
necessary distinguishing examples for witnessing $x\sims y$,
$(u,v)\napproxsbot$, or $(u,v)\napproxs (u',v')$ then this also holds for every
$S'$ that is consistent with $L$ and "extends" $S$. 

\begin{rem}~\label{rem:poly-pairs}
  For $x\sims y$ we need $4(|\Rsim|-1)$ pairs in $S$. For $(u,v) \napproxbot$ we need pairs for $uvx \sims u$ and one additional pair for $(u,vx)\ins L$. For $(u,v)\inapproxs (u',v')$ we either need two pairs for witnessing $u\nsims u'$, or $uv\nsims  u'v'$, or we need the pairs for $uvx\sims u$ and additionally two pairs for the membership tests.
\end{rem}

\begin{rem}
  We manage to have a "stabilizing" test $x\sims y$ thanks to the enumeration $\Rsim$ of
  all equivalence classes of $\sims$ relation. 
  It is tempting to follow the same route and get a test for $(u,v)\approxl
  (u',v')$ by enumerating all classes of $\approxl$ relation. 
  The problem is that there are exponentially many such classes with respect to
  the size of the minimal automaton, cf. Example~\ref{ex:ak-2}, so we cannot do this if we want
  polynomial size samples and algorithms. 
\end{rem}

Now we consider line~\ref{alg:lineE} of Listing~\ref{lst:algo} and implement it as a function
$\FindNT(R,S)$ given in Listing~\ref{lst:find-NT}.
The function when called with $\Rsim$ returns the set $\NTsim$ of all $u\in \Rsim$
such that there is $x$ with $ux^\w\ins L$ and $ux\sims u$. 
Thanks to $\FindR(S)$ we can assume that $\FindNT$ has access to $\Rsim$, simply
because $S$ can be big enough for $\FindR$ to return the correct $\Rsim$.
Then in the rest of the algorithm, the $x \sims y$ test that uses $\Rsim$ is
equivalent to the $x\siml y$ test. 

\begin{lstlisting}[caption={Finding $\NTsim$},label=lst:find-NT,float,frame=lines,escapechar=|]
function $\iFindNT(R,S)$:
  $NT=\es$
  forevery $u\in R$
      if there is $x$ with $ux^\w\ins L$ and $ux\sims u$ then |\label{alg:findE}|
      $NT:=NT\cup\set{u}$
  return($NT$)
\end{lstlisting}

\begin{lem}
  There is a "stabilizing" $S$ of size polynomial in $|\Aae|$ such that $\NTsim=\FindNT(\Rsim,S)$.
  Moreover, the algorithm runs in time polynomial in $S$.
\end{lem}
\begin{proof}
  Simply $S$ should be big enough to have all the witnesses required in
  line~\ref{alg:findE} of Listing~\ref{lst:find-NT} for every $u\in\Rsim$. 
  The test $ux^\w\ins L$ requires the pair $(u,x)\in S^+$.
  The test $ux\sims u$ requires a number of pairs linear in $|\Rsim|$ pairs, as
  we have noted in Remark~\ref{rem:poly-pairs}.
\end{proof}
Next, in Listing~\ref{lst:automaton} we construct an automaton for the given set
of components.
We can use $\sims$ tests because we can assume that $R=\Rsim$. 
We need $\Oo(|\Rsim|)$ elements in $S$ for each test.
Using this test, the function $\automaton$ simply adds the rank $1$ transitions to the components, and defines the initial states. It is clear that it works in polynomial time.

\begin{lstlisting}[caption={Constructing automaton for $\Cc$},label=lst:automaton,float,frame=lines]
function $\iautomaton(R,\Cc,S)$:
  $Q^\Aa=\bigcup\set{Q : \struct{Q,\D}\in \Cc}$
  $\D^\Aa_2=\bigcup\set{\D : \struct{Q,\D}\in \Cc}$
  $\D^\Aa_1=\{(u,v)\aact{a:1}(u',v'):a\in\S,\ (u,v),\ (u',v')\in Q^\Aa,\ uva\sims u'v'\}$
  $\init=\set{(u,v)\in Q^\Aa : uv\sims \e}$
  return($\struct{Q^\Aa,\D^\Aa_1\cup\D^\Aa_2,\init}$)
\end{lstlisting}

The following step is the test in line~\ref{palg:while} of Listing~\ref{lst:palgo}.
It is realized by direct inclusion test $L\neqs L(\Aa)$ that we have defined
above.
This step requires iterating over all elements from the "sample" and to check if the positive examples are accepted by $\Aa$, and the negative ones are rejected by $\Aa$, which can be done in polynomial time.

Next, we come to line~\ref{palg:u} of Listing~\ref{lst:palgo}.
The function call implements finding a $u$ from the same line of Listing~\ref{lst:algo}.
The code for this function is given in Listing~\ref{lst:findu}. 
It is a simple search for $u_i$ such that there is $(u_i,x) \ins L$ with $u_ix\sims u_i$ and
$u_ix^\w\notin L(\Aa)$.

\begin{lstlisting}[caption={Finding $u_i$},label=lst:findu,float,frame=lines]
function $\iFindu(\Aa,R,S)$:
  Let $R=\set{u_1,\dots,u_k}$
  $i:=0$
  repeat
    $i=i+1$
    found:=false
    forevery $(u_i,x) \ins L$  do
      if $u_ix\sims u_i$ and $u_ix^\w\not\in L(\Aa)$ then found:=true
  until found or $i=k$
  if found then return($u_i$) else return($\bot$)
\end{lstlisting}

\begin{lem}
  For every $\Aa$ there is a "stabilizing" $S$ of size polynomial in
  $|\Aa|+|\Aae|$ such that $u_i=\Findu(\Aa,\Rsim,S)$ is the same as in 
  line~\ref{alg:u} of Listing~\ref{lst:algo}. Moreover, $\Findu(\Aa,\Rsim,S)$ runs in time
  polynomial in the size of $\Aa$, $\Rsim$ and $S$. 
\end{lem}
\begin{proof}
  The algorithm uses only operations that are stable under extension. 
  The for-loop of the algorithm runs in polynomial time as the
  number of possible $x$ is bounded by the size of $S^+$, and the test $u_ix^\w\not\in L(\Aa)$ is polynomial in $u_i$, $x$, and $\Aa$.
  A "sample" $S$ that is big enough to have correct $\Rsim$ and contains some $(u_i,x) \in S^+$ for the $u_i$  from line~\ref{alg:u}  of the idealized algorithm in Listing~\ref{lst:algo}, gives the correct $u_i$.
\end{proof}

The next step is finding $x$ from line~\ref{palg:x}. 
This is the  "smallest" $x$ such that $u_ix\sims u_i$ and $u_ix^\w\not\in L(\Aa)$.
This is realized by the function in Listing~\ref{lst:findx}.
\begin{lstlisting}[caption={Finding $x$},label=lst:findx,float,frame=lines]
function $\iFindx(u,\Aa,R,S)$:
  $D:=\set{y : (u,y)\ins L, uy\sims u}$
  $x:=\min\set{y\in D: uy^\w\not\in L(\Aa)}$ // $x=\bot$ if no such $y$
  return($x$)
\end{lstlisting}

\begin{lem}
  There is a "stabilizing" $S$ of polynomial size in $|\Aa|+|\Aae|$ such that $x=\Findx(u,\Aa,\Rsim,S)$ is as in
  line~\ref{alg:x} of Listing~\ref{lst:algo}, provided that $u$ equals $u_i$ as computed in line~\ref{alg:u} of Listing~\ref{lst:algo}. Moreover, $\Findx(u,\Aa,\Rsim,S)$ runs in time
  polynomial in the size of its arguments. 
\end{lem}
\begin{proof}
  The size of $D$ is bounded by the size of $S$.
  For the algorithm to return the correct answer, it is enough that it has the
  pair $(u_i,x)$ in $S^+$.
\end{proof}

The $d$ in line~\ref{palg:d} of Listing~\ref{lst:palgo} is defined in the same
way as $d$ in the same line of Listing~\ref{lst:algo}.
Constructing a component $C(u,v)$ from line~\ref{alg:C} of
Listing~\ref{lst:algo} is done in line~\ref{palg:C} of Listing~\ref{lst:palgo}
by the function $\construct$ presented in Listing~\ref{lst:constr}.
While its pseudo code is quite long, it is a simple application of
the definition using what we have already computed. If the "sample" contains enough distinguishing examples, then the set $K$ that is computed corresponds to $\Rappl(u,v)$. This is similar to the computation of $\Rsim$.

\begin{lstlisting}[caption={Constructing $C(u,v)$},label=lst:constr,float,frame=lines]
function $\iconstruct(u,v,R,S)$:
  $K:=\set{\e}$
  repeat
    $D:=\set{x : (u,vx) \napproxsbot\text{ and }\forall y\in K.\ (u,vx)\napproxs (u,vy)}$
    if $D\not=\es$ then add $\min(D)$ to $K$
  until $D=\es$
  $T=\{(x,a,y)\in K\times\S\times K:$ 
         $\forall z\in K.\ z\not=y\imp (u,vxa)\napproxs(u,vz)\}$
  $Q=\set{(u,vx): x\in K}$
  $\D=\set{(u,vx)\act{a}(u,vy) : (x,a,y)\in T}$
  return($\struct{Q,\D}$)
\end{lstlisting}

The last missing step is the implementation of the function $\FindPointedL(u,v)$ from
Listing~\ref{lst:find-pointed}. Actually it needs only minimal changes to be
implemented in the passive setting: the tests with the subscript $L$ are changed to
tests with the subscript $S$.
The implementation is called $\FindPointedS(u,v,R,S)$ and is shown in
Listing~\ref{lst:find-pointed-passive}.

\begin{lstlisting}[caption={Finding a "pointed" $(u,w\bv)$ extending $(u,v)$},label=lst:find-pointed-passive,float,frame=lines,escapechar=|]
function $\iFindPointedS(u,v,R,S)$:
  Set $\bv=v$.
  while $\exists x$ s.t. $u\sims u\bv x$, $(u,\bv x \bv)\napproxsbot$ and $u(\bv x)^\w\nins L$
    Find the $\text{\kl{smallest}}$ such $x$. |\label{line:find-x-i}|
    $\bv:=\bv(x\bv)^m$, where $m$ the biggest s.t., $(u,\bv(x\bv)^m)\napproxsbot$. |\label{line:find-m}|
  Set $w=\e$.
  while $\exists x$ s.t. $u\sims uw\bv x$, $(u,w\bv x\bv)\napproxs\set{\bot,(u,w\bv)}$
    Find the $\text{\kl{smallest}}$ such $x$. |\label{line:find-x-ii}|
    $w:=w\bv x$.
  return $(u,w\bv)$  
\end{lstlisting}

\begin{lem}
  There is a "stabilizing" $S$ of polynomial size in $|\Aae|$ such that
  $(u,w\bv)=\FindPointedS(u,v,\Rsim,S)$ is the same as $\FindPointedL(u,v)$
  from Listing~\ref{lst:find-pointed}. 
  Moreover, $\FindPointedS(u,v,\Rsim,S)$ runs in time polynomial in the size of
  its arguments. 
\end{lem}
\begin{proof}
  Observe that the two loops run in polynomial time, as the number of iterations
  is bounded by the length of the longest second component of a pair in $S$. For the first loop this follows from the fact that the length $\bv$ grows in each iteration of the loop, and in the second loop $w$ grows with each iteration. Once the length of $\bv$, respectively $w\bv$, has exceeded the length of a longest second component in $S$, there cannot be a further $x$ satisfying the respective loop condition, because this would require pairs in $S$ whose second component contains $\bv$, respectively $w\bv$, as prefix.

  The tests $\sims$, $\napproxsbot$, and $\napproxs$ used in
  $\FindPointedS$ are stable under extensions and give the same results as the
  corresponding test with $L$ subscript provided $S$ contains enough samples.
  Proposition~\ref{prop:idealized-pointed-finding-correct} states that the
  result of $\FindPointedL$ is of size polynomial in  $|\Aae|$. 
  This implies that witnesses $x$ required by $\FindPointedS$ in
  lines~\ref{line:find-x-i} and~\ref{line:find-x-ii} exist and are of
  size polynomial in $|\Aae|$.  
  Moreover, every value of $\bv$ in the execution of the algorithm is also of
  size polynomial in $|\Aae|$. 
  For the algorithm to give the correct result the sample should have enough
  information about the witnesses.
  For every such witness $x$ from line~\ref{line:find-x-i} we need in $S$
  enough information for it to pass three tests.
  In particular, in line~\ref{line:find-m} we need only enough information in
  $S$ for the correct value of $m$, meaning the value algorithm 
  $\FindPointedL$ would use. 
  For the other values of $m$, if there is not enough information then the test will
  simply fail.
  By Remark~\ref{rem:poly-pairs} each test requires the presence in $S$ of a linear
  number of pairs wrt.\ the size of $|\Rsim|$.
  For every witness $x$ from line~\ref{line:find-x-ii} we also need witnesses
  for three tests.
  To summarize, in order for $\FindPointedS$ to give a correct answer only a
  polynomial number of pairs in $S$ is needed, and each needed pair is of
  polynomial length in $|\Aae|$.
\end{proof}

At this point we have all the ingredients for a proof of
Theorem~\ref{thm:learner}.

\noindent\emph{Proof of Theorem~\ref{thm:learner}:}
Consider $S$ consistent with $L$.
Suppose $S$ is stabilizing for every call of every function in the execution of
the algorithm $\Learn(S)$ from Listing~\ref{lst:palgo}, and moreover $S$
guarantees that the returned value is correct in the sense stated in the lemmas above.
Then the algorithm behaves exactly the same as the idealized algorithm from
Listing~\ref{lst:algo}. 
In particular, the size of $\Aa$ being constructed is not bigger than that of
$\Aae$.
This means that the size of $S$ needs not to be bigger than some polynomial in
the size of $\Aae$ for the algorithm $\Learn(S)$ to return the same result as the
idealized algorithm. 
As the idealized algorithm is correct thanks to Theorem~\ref{thm:idealized},
so is the result of $\Learn(S)$. 
Finally, thanks to the lemmas in this section, the running time of $\Learn(S)$
is guaranteed to be polynomial in the size of the sample $S$.

\section{Conclusion} \label{sec:conclusion}
The main technical result of this paper is a passive learning algorithm for
history-deterministic co-Büchi automata that is polynomial in both time and
data. This is the first such algorithm for a class that remains non-trivial when
restricted to prefix-independent $\w$-languages. With Example~\ref{ex:long-SCC-path},
we have highlighted the challenges concerning polynomial learning of $\w$-automata,
that can in part explain the slow progress in this area in the past.
The straightforward congruence-based approach encounters the problem of an
exponential number of congruence classes, and there are classes that can only be
identified with exponentially long examples. 

The other foundational contribution is a congruence-based description of minimal
history-deterministic co-Büchi automata. The main novelty here is the notion of
a "pointed" pair. Another subtlety is the definition of the $\tbe$ relation, which
is similar to those present in the literature but does not require the first
components to be $\siml$ equivalent. 

For future work, it would be interesting to go beyond the co-Büchi
condition. A good starting point are the \emph{chains of co-Büchi automata (COCOA)}
from \cite{EhlersS22}, a canonical representation of all
$\w$-regular languages as a Boolean combination of history-deterministic
co-Büchi automata. Another goal is to develop an active, Angluin-style~\cite{Angluin87} learning
algorithm for the co-Büchi class.
Finally, it would be interesting to have faster minimization algorithms for
history-deterministic co-Büchi automata.
As in the case of finite words, understanding their structure in terms of
congruences may be helpful.

\bibliographystyle{alphaurl}

\bibliography{references}

@article{Abu.Ehl.NPhard2025,
  author       = {Bader {Abu Radi} and
                  R{\"{u}}diger Ehlers},
  title        = {Characterizing the Polynomial-Time Minimizable {\(\omega\)}-Automata},
  journal      = {CoRR},
  volume       = {abs/2504.20553},
  year         = {2025},
  url          = {https://doi.org/10.48550/arXiv.2504.20553},
  doi          = {10.48550/ARXIV.2504.20553},
  eprinttype    = {arXiv},
  eprint       = {2504.20553},
  bibsource    = {dblp computer science bibliography, https://dblp.org}
}

@inproceedings{Ehl.Kha.Fully2024,
  title = {Fully {{Generalized Reactivity}}(1) {{Synthesis}}},
  booktitle = {TACAS, Tools and {{Algorithms}} for the {{Construction}} and {{Analysis}} of {{Systems}}},
  author = {Ehlers, R{\"u}diger and Khalimov, Ayrat},
  editor = {Finkbeiner, Bernd and Kov{\'a}cs, Laura},
  year = {2024},
  volume = {14570},
  pages = {83--102},
  doi = {10.1007/978-3-031-57246-3_6},
  urldate = {2024-10-25},
  isbn = {978-3-031-57245-6 978-3-031-57246-3}
}

@inproceedings{Klarlund94,
  author    = {Nils Klarlund},
  title     = {A Homomorphism Concepts for omega-Regularity},
  booktitle = {Computer Science Logic, 8th International Workshop, {CSL} '94, Kazimierz,
               Poland, September 25-30, 1994, Selected Papers},
  pages     = {471--485},
  year      = {1994},
  url       = {https://doi.org/10.1007/BFb0022276},
  doi       = {10.1007/BFb0022276},
  series    = {Lecture Notes in Computer Science},
  volume    = {933},
  publisher = {Springer}
}

@article{MalerS97,
  author = {Oded Maler and Ludwig Staiger},
  title = {On syntactic congruences for {$\omega$}-languages},
  pages = {93-112},
  journal = {Theoretical Computer Science},
  year = 1997,
  volume = 183,
  number = 1
}

@inproceedings{AngluinFS20,
  author    = {Dana Angluin and Dana Fisman and Yaara Shoval},
  title     = {Polynomial Identification of \(\omega\)-Automata},
  booktitle = {Tools and Algorithms for the Construction and Analysis of
               Systems - 26th International Conference, {TACAS} 2020},
  series    = {Lecture Notes in Computer Science},
  volume    = {12079},
  pages     = {325--343},
  publisher = {Springer},
  year      = {2020},
  url       = {https://doi.org/10.1007/978-3-030-45237-7\_20},
  doi       = {10.1007/978-3-030-45237-7\_20},
  bibsource = {dblp computer science bibliography, https://dblp.org}
}

@book{BaierK2008,
  title     = {Principles of model checking},
  author    = {Baier, Christel and Katoen, Joost-Pieter},
  publisher = {MIT Press},
  year      = {2008}
}

@article{Angluin87,
  author    = {Dana Angluin},
  title     = {Learning Regular Sets from Queries and Counterexamples},
  journal   = {Inf. Comput.},
  volume    = {75},
  number    = {2},
  year      = {1987},
  pages     = {87--106},
  ee        = {http://dx.doi.org/10.1016/0890-5401(87)90052-6},
  bibsource = {DBLP, http://dblp.uni-trier.de}
}

@article{MalerP95,
  author    = {Oded Maler and Amir Pnueli},
  title     = {On the Learnability of Infinitary Regular Sets},
  journal   = {Inf. Comput.},
  volume    = {118},
  number    = {2},
  pages     = {316--326},
  year      = {1995},
  url       = {https://doi.org/10.1006/inco.1995.1070},
  doi       = {10.1006/inco.1995.1070},
  bibsource = {dblp computer science bibliography, https://dblp.org}
}

@InProceedings{Buchi62,
  author= {J. Richard B{\"u}chi},
  title= {On a decision method in restricted second order arithmetic},
  booktitle= {International Congress on Logic, Methodology and Philosophy of Science},
  pages= {1--11},
  year= {1962},
  publisher= {Stanford University Press},
  location= {Stanford, USA}
}

@article{LuttenbergerMS20,
  author       = {Michael Luttenberger and
                  Philipp J. Meyer and
                  Salomon Sickert},
  title        = {Practical synthesis of reactive systems from {LTL} specifications
                  via parity games},
  journal      = {Acta Informatica},
  volume       = {57},
  number       = {1-2},
  pages        = {3--36},
  year         = {2020},
  url          = {https://doi.org/10.1007/s00236-019-00349-3},
  doi          = {10.1007/S00236-019-00349-3},
  bibsource    = {dblp computer science bibliography, https://dblp.org}
}

@inproceedings{BohnL21,
  author    = {Le{\'{o}}n Bohn and Christof L{\"{o}}ding},
  title     = {Constructing Deterministic {\(\omega\)}-Automata from Examples by
               an Extension of the {RPNI} Algorithm},
  booktitle = {46th International Symposium on Mathematical Foundations of
               Computer Science, {MFCS} 2021},
  pages     = {20:1--20:18},
  year      = {2021},
  url       = {https://doi.org/10.4230/LIPIcs.MFCS.2021.20},
  doi       = {10.4230/LIPIcs.MFCS.2021.20},
  series    = {LIPIcs},
  volume    = {202},
  publisher = {Schloss Dagstuhl - Leibniz-Zentrum f{\"{u}}r Informatik}
}

@inproceedings{BohnL22,
  author    = {Le{\'{o}}n Bohn and Christof L{\"{o}}ding},
  title     = {Passive Learning of Deterministic B{\"{u}}chi Automata by
               Combinations of {DFA}s},
  booktitle = {49th International Colloquium on Automata, Languages, and
               Programming, {ICALP} 2022, July 4-8, 2022, Paris, France},
  pages     = {114:1--114:20},
  year      = {2022},
  doi       = {10.4230/LIPIcs.ICALP.2022.114},
  series    = {LIPIcs},
  volume    = {229},
  publisher = {Schloss Dagstuhl - Leibniz-Zentrum f{\"{u}}r Informatik}
}

@article{BohnL24,
  author       = {Le{\'{o}}n Bohn and
                  Christof L{\"{o}}ding},
  title        = {Constructing Deterministic Parity Automata from Positive and Negative
                  Examples},
  journal      = {TheoretiCS},
  volume       = {3},
  year         = {2024},
  url          = {https://doi.org/10.46298/theoretics.24.17},
  doi          = {10.46298/THEORETICS.24.17},
  bibsource    = {dblp computer science bibliography, https://dblp.org}
}

@article{Gold78,
  author    = {E. Mark Gold},
  title     = {Complexity of Automaton Identification from Given Data},
  journal   = {Information and Control},
  volume    = {37},
  number    = {3},
  pages     = {302--320},
  year      = {1978},
  url       = {https://doi.org/10.1016/S0019-9958(78)90562-4},
  doi       = {10.1016/S0019-9958(78)90562-4},
  bibsource = {dblp computer science bibliography, https://dblp.org}
}

@incollection{Kup.Skr.Determinisation2015,
  title = {On {{Determinisation}} of {{Good-for-Games Automata}}},
  booktitle = {Automata, {{Languages}}, and {{Programming}}},
  author = {Kuperberg, Denis and Skrzypczak, Michal},
  editor = {Halldorsson, Magnus M. and Iwama, Kazuo and Kobayashi, Naoki and Speckmann, Bettina},
  year = {2015},
  volume = {9135},
  pages = {299--310},
  publisher = {Springer Berlin Heidelberg},
  doi = {10.1007/978-3-662-47666-6_24},
  urldate = {2024-06-16},
  isbn = {978-3-662-47665-9 978-3-662-47666-6}
}

@article{Rad.Kup.Minimization2022,
  title = {Minimization and {{Canonization}} of {{GFG Transition-Based Automata}}},
  author = {Radi, Bader Abu and Kupferman, Orna},
  year = {2022},
  journal = {Logical Methods in Computer Science},
  volume = {Volume 18, Issue 3},
  pages = {7587},
  issn = {1860-5974},
  doi = {10.46298/lmcs-18(3:16)2022},
  urldate = {2024-06-16},
  copyright = {https://creativecommons.org/licenses/by/4.0}
}

@inproceedings{Ver.Ham.Flexfringe2017,
  title = {Flexfringe: {{A Passive Automaton Learning Package}}},
  shorttitle = {Flexfringe},
  booktitle = {2017 {{IEEE International Conference}} on {{Software Maintenance}} and {{Evolution}} ({{ICSME}})},
  author = {Verwer, Sicco and Hammerschmidt, Christian A.},
  year = {2017},
  pages = {638--642},
  publisher = {IEEE},
  doi = {10.1109/ICSME.2017.58},
  urldate = {2024-06-15},
  isbn = {978-1-5386-0992-7}
}

@InProceedings{Schewe10,
  author =	{Sven Schewe},
  title =	{{Beyond Hyper-Minimisation---Minimising DBAs and DPAs is NP-Complete}},
  booktitle =	{IARCS Annual Conference on Foundations of Software Technology and Theoretical Computer Science (FSTTCS 2010)},
  pages =	{400--411},
  series =	{Leibniz International Proceedings in Informatics (LIPIcs)},
  ISBN =	{978-3-939897-23-1},
  ISSN =	{1868-8969},
  year =	{2010},
  volume =	{8},
  editor =	{Kamal Lodaya and Meena Mahajan},
  publisher =	{Schloss Dagstuhl--Leibniz-Zentrum fuer Informatik},
  address =	{Dagstuhl, Germany},
  URL =		{http://drops.dagstuhl.de/opus/volltexte/2010/2881},
  URN =		{urn:nbn:de:0030-drops-28816},
  doi =		{10.4230/LIPIcs.FSTTCS.2010.400}
}

@inproceedings{Casares22,
  author       = {Antonio Casares},
  title        = {On the Minimisation of Transition-Based Rabin Automata and the Chromatic
                  Memory Requirements of Muller Conditions},
  booktitle    = {30th {EACSL} Annual Conference on Computer Science Logic, {CSL} 2022,
                  February 14-19, 2022, G{\"{o}}ttingen, Germany (Virtual Conference)},
  pages        = {12:1--12:17},
  series       = {LIPIcs},
  volume       = {216},
  publisher    = {Schloss Dagstuhl - Leibniz-Zentrum f{\"{u}}r Informatik},
  year         = {2022},
  url          = {https://www.dagstuhl.de/dagpub/978-3-95977-218-1},
  isbn         = {978-3-95977-218-1},
  bibsource    = {dblp computer science bibliography, https://dblp.org}
}

@inproceedings{HenzingerP06,
  author       = {Thomas A. Henzinger and
                  Nir Piterman},
  title        = {Solving Games Without Determinization},
  booktitle    = {Computer Science Logic, 20th International Workshop, {CSL} 2006, 15th
                  Annual Conference of the EACSL, Szeged, Hungary, September 25-29,
                  2006, Proceedings},
  pages        = {395--410},
  series       = {Lecture Notes in Computer Science},
  volume       = {4207},
  publisher    = {Springer},
  year         = {2006},
  url          = {https://doi.org/10.1007/11874683},
  doi          = {10.1007/11874683},
  isbn         = {3-540-45458-6},
  bibsource    = {dblp computer science bibliography, https://dblp.org}
}

@article{BokerL23,
  author       = {Udi Boker and
                  Karoliina Lehtinen},
  title        = {When a Little Nondeterminism Goes a Long Way: An Introduction to History-Determinism},
  journal      = {{ACM} {SIGLOG} News},
  volume       = {10},
  number       = {1},
  pages        = {24--51},
  year         = {2023},
  url          = {https://doi.org/10.1145/3584676.3584682},
  doi          = {10.1145/3584676.3584682},
  bibsource    = {dblp computer science bibliography, https://dblp.org}
}

@inproceedings{Thomas09,
author		= {Thomas, Wolfgang},
title		= {Facets of Synthesis: Revisiting {Church}'s Problem},
booktitle	= {Proceedings of the 12th International Conference on Foundations of Software Science and Computational Structures, FOSSACS 2009},
series		= {Lecture Notes in Computer Science},
volume		= 5504,
publisher	= {Springer},
pages		= {1--14},
year            = {2009},
}

@InCollection{Thomas90,
  author = 	 {Wolfgang Thomas},
  title = 	 {Automata on Infinite Objects},
  booktitle = 	 {Handbook of Theoretical Computer Science},
  pages =	 {133--192},
  publisher =	 {Elsevier Science Publishers},
  year =	 1990,
  volume =	 {B: Formal Models and Semantics},
  address =	 {Amsterdam}
}

@Article{Loding01,
  author = 	 {L{\"o}ding, Christof},
  title = 	 {Efficient Minimization of Deterministic Weak $\omega$-Automata},
  journal =      {Information Processing Letters},
  volume =       79,
  number =       3,
  pages =        {105--109},
  publisher =    {Elsevier Science Publishers},
  year =         2001,
}

@article{Staiger83,
  author = {Ludwig Staiger},
  title = {Finite-State $\omega$-Languages},
  journal = {Journal of Computer and System Sciences},
  volume = 27,
  year = 1983,
  pages = {434-448}
}

@article{MichaliszynO22,
  author       = {Jakub Michaliszyn and
                  Jan Otop},
  title        = {Learning infinite-word automata with loop-index queries},
  journal      = {Artif. Intell.},
  volume       = {307},
  pages        = {103710},
  year         = {2022},
  url          = {https://doi.org/10.1016/j.artint.2022.103710},
  doi          = {10.1016/J.ARTINT.2022.103710},
  bibsource    = {dblp computer science bibliography, https://dblp.org}
}

@article{BiermannF72,
  author    = {Alan W. Biermann and Jerome A. Feldman},
  title     = {On the Synthesis of Finite-State Machines from Samples of Their
               Behavior},
  journal   = {{IEEE} Trans. Computers},
  volume    = {21},
  number    = {6},
  pages     = {592--597},
  year      = {1972},
  url       = {https://doi.org/10.1109/TC.1972.5009015},
  doi       = {10.1109/TC.1972.5009015},
  bibsource = {dblp computer science bibliography, https://dblp.org}
}

@book{TrakhtenbrotB73,
  author    = {Boris A. Trakhtenbrot and Y.M. Barzdin},
  title     = {Finite Automata: Behavior and Synthesis},
  publisher = {North-Holland Publishing Company},
  address   = {Amsterdam},
  year      = 1973
}

@incollection{LopezG16,
  author    = {Damian L{\'o}pez and Pedro Garc{\'i}ía},
  title     = {On the Inference of Finite State Automata from Positive and
               Negative Data},
  booktitle = {Topics in Grammatical Inference},
  publisher = {Springer},
  year      = {2016},
  editor    = {In: Heinz J., Sempere J.},
  doi       = {https://doi.org/10.1007/978-3-662-48395-4_4}
}

@Book{HopcroftU79,
  author = 	 {Hopcroft, John E. and Ullman, Jeffrey D.},
  title = 	 {Introduction to Automata Theory, Languages, and
                  Computation},
  publisher = 	 {Addison Wesley},
  year = 	 {1979}
}

@inproceedings{GarciaO92,
  author     = {Oncina, Jos\'{e} and Garc\'{i}a, Pedro},
  title      = {Identifying regular languages in polynomial time},
  booktitle  = {Proceedings of the International Workshop on
                  Structural and Syntactic Pattern Recognition},
  series     = {Machine Perception and Artificial Intelligence},
  volume     = {5},
  year       = {1992},
  publisher  = {World Scientific},
  pages      = {99–-108},
}

@inproceedings{MannaP89,
  author       = {Zohar Manna and
                  Amir Pnueli},
  title        = {A Hierarchy of Temporal Properties},
  booktitle    = {Proceedings of the Ninth Annual {ACM} Symposium on Principles of Distributed
                  Computing, Quebec City, Quebec, Canada, August 22-24, 1990},
  pages        = {377--410},
  publisher    = {{ACM}},
  year         = {1990},
  url          = {http://dl.acm.org/citation.cfm?id=93385},
  isbn         = {0-89791-404-X},
  bibsource    = {dblp computer science bibliography, https://dblp.org}
}

@inproceedings{EhlersS22,
  author    = {R{\"{u}}diger Ehlers and
               Sven Schewe},
  title     = {Natural Colors of Infinite Words},
  booktitle = {Foundations of Software Technology
               and Theoretical Computer Science, {FSTTCS} 2022},
  pages     = {36:1--36:17},
  year      = {2022},
  doi       = {10.4230/LIPIcs.FSTTCS.2022.36},
  series    = {LIPIcs},
  volume    = {250},
  publisher = {Schloss Dagstuhl - Leibniz-Zentrum f{\"{u}}r Informatik}
}





\end{document}